\documentclass[a4paper,leqno]{amsart}

\usepackage{latexsym}
\usepackage[english]{babel}
\usepackage{fancyhdr}
\usepackage[mathscr]{eucal}
\usepackage{amsmath}
\usepackage{mathrsfs}
\usepackage{amsthm}
\usepackage{amsfonts}
\usepackage{amssymb}
\usepackage{amscd}
\usepackage{bbm}
\usepackage{graphicx}
\usepackage{graphics}
\usepackage{color}
\usepackage{enumerate}
\usepackage{framed}

\newcommand{\ud}{\mathrm{d}}

\newcommand{\cH}{\mathcal{H}}

\newcommand{\N}{\mathbb{N}}

\newcommand{\R}{\mathbb{R}}
\newcommand{\C}{\mathbb{C}}

\newcommand{\la}{\langle}
\newcommand{\ra}{\rangle}
\newcommand{\setl}[1]{\left\{#1\right\}}
\newcommand{\norm}[1]{\|#1\|}
\renewcommand{\Im}{\mathrm{Im}}
\renewcommand{\Re}{\mathrm{Re}}

\theoremstyle{plain}
\newtheorem{theorem}{Theorem}[section]
\newtheorem{lemma}[theorem]{Lemma}

\newtheorem{prop}[theorem]{Proposition}

\theoremstyle{definition}

\newtheorem*{remark*}{Remark}

\numberwithin{equation}{section}

\begin{document}

\allowdisplaybreaks

\title[Mean-field dynamics for condensates via Fock space methods]{Mean-field dynamics for mixture condensates via Fock space methods}
\author[G.~de Oliveira]{Gustavo de Oliveira}
\address[G.~de Oliveira]{Departamento de Matem\'atica\\ICEx\\Universidade Federal de Minas Gerais\\Av. Ant\^onio Carlos, 6627\\31270-901 Belo Horizonte--MG (Brazil)}
\email{goliveira5@ufmg.br}
\author[A.~Michelangeli]{Alessandro Michelangeli}
\address[A.~Michelangeli]{International School for Advanced Studies -- SISSA\\via Bonomea 265\\34136 Trieste (Italy)}
\email{alemiche@sissa.it}

\begin{abstract}
  We consider a mean-field model to describe the dynamics of $N_1$ bosons of species one and $N_2$ bosons of species two in the limit as $N_1$ and $N_2$ go to infinity.
  We embed this model into Fock space and use it to describe the time evolution of coherent states which represent two-component condensates.
  Following this approach, we obtain a microscopic quantum description for the dynamics of such systems, determined by the Schr\"{o}dinger equation.
  Associated to the solution to the Schr\"{o}dinger equation, we have a reduced density operator for one particle in the first component of the condensate and one particle in the second component.
  In this paper, we estimate the difference between this operator and the projection onto the tensor product of two functions that are solutions of a system of equations of Hartree type.
  Our results show that this difference goes to zero as $N_1$ and $N_2$ go to infinity.
\end{abstract}

\date{\today}

\thanks{Partially supported by the 2014-2017 MIUR-FIR grant \emph{``Cond-Math: Condensed Matter and Mathematical Physics''} code RBFR13WAET, and by a 2015 visiting research fellowship at the International Center for Mathematical Research (CIRM), Trento.}

\maketitle

\section{Introduction}

One of the typical experimental realisations of Bose-Einstein condensation (also called BEC henceforth) concerns the preparation and the dynamical evolution of the so-called mixture condensates \cite{Hall2008_multicompBEC_experiments,Malomed2008_multicompBECtheory}. They consist of a Bose gas formed by two different species of bosons, with both intra-species and inter-species interactions, which exhibits condensation in each component. This results in a macroscopic occupation of some one-body orbital for the first type of bosons and another one for the second type. Such systems are customarily prepared as a gas of atoms of the same element, typically $^{87}\mathrm{Rb}$, which occupy two hyperfine states \cite{MBGCW-1997,Matthews_HJEWC_DMStringari_PRL1998,HMEWC-1998,Hall-Matthews-Wieman-Cornell_PRL81-1543}, or also as heteronuclear mixtures such as $^{41}\mathrm{K}$-$^{87}\mathrm{Rb}$ \cite{Modugno-Ferrari-Inguscio-etal-Science2001_multicompBEC}, $^{41}\mathrm{K}$-$^{85}\mathrm{Rb}$ \cite{Modugno-PRL-2002}, $^{39}\mathrm{K}$-$^{85}\mathrm{Rb}$ \cite{MTCBM-PRL2004_BEC_heteronuclear} and $^{85}\mathrm{K}$-$^{87}\mathrm{Rb}$ \cite{Papp-Wieman_PRL2006_heteronuclear_RbRb}. In the former case, we refer to the part of the experiment in which no interconversion between particles of different hyperfine states occur. For a comprehensive review of the physical properties of mixture condensates, we refer to \cite[Chapter 21]{pita-stringa-2016}.

A mixture condensate (with a fixed number of particles) is naturally modeled as a many-body system of $N_1$ indistinguishable bosons of the first species and $N_2$ indistinguishable bosons of the second species.
The Hilbert space for the system is the tensor product 
\begin{equation}\label{eq:Hspace_for_mixture}
 \cH = L^2_s(\R^{3N_1}) \otimes L^2_s(\R^{3N_2}),
\end{equation}
where $L^2_s(\R^{3N})$ denotes the subspace of functions of $L^2(\R^{3N})$ that are symmetric with respect to permutation of any pair of variables (which correspond to particles).
The mean-field Hamiltonian of the system has the form
\[
  H_{N_1,N_2} = h_{N_1} \otimes I + I \otimes h_{N_2} + \mathcal{V}_{N_1,N_2},
\]
where, for $p = 1, 2$,
\[
  h_{N_p} = \sum_{j=1}^{N_p} -\Delta_{x_j} + \frac{1}{N_p} \sum_{1 \le j < k \le N_p} V_p(x_j-x_k)
\]
and
\[
  \mathcal{V}_{N_1,N_2} = \frac{1}{N_1+N_2} \sum_{j=1}^{N_1} \sum_{k=1}^{N_2} V_{12}(x_j-y_k).
\]
Here, $x_j \in \R^3$ represents the $j$th variable corresponding to the first factor of $\mathcal{H}$ and $y_k \in \R^3$ represents the $k$th variable corresponding to the second factor.
We will describe later the hypothesis on the interaction potentials $V_1$, $V_2$ and $V_{12}$.

We will consider a mixture condensate with variable number of particles of each species, where particles can not switch species.
In order to do this, we will embed the model mentioned above into a tensor product of Fock spaces, as described in Section \ref{s:fock}.

The $N_1,N_2$-dependent factors in the Hamiltonian $H_{N_1,N_2}$ are typical of the mean-field regime.
When $N_1 \to \infty$ and $N_2\to\infty$ with $N_1/N_2 \to \text{constant}$, the factors guarantee that the kinetic and potential terms of the Hamiltonian remain comparable and thus the many-body dynamics remains non-trivial in the limit.
In fact, there are $N_1+N_2$ kinetic terms and $\frac{1}{2}(N_1+N_2)(N_1+N_2-1)$ potential terms in $H_{N_1,N_2}$, but the mean-field factors reduce the order of the potential energy to
\[
\frac{1}{N_1}\cdot\frac{N_1(N_1-1)}{2}+\frac{1}{N_2}\cdot\frac{N_2(N_2-1)}{2}+\frac{1}{N_1+N_2}\cdot N_1\, N_2 =O(N_1+N_2).
\]
One could obtain a similar regime, for example, by choosing the factor $1/(N_1+N_2)$ for all potential terms, or by replacing the factor in $\mathcal{V}_{N_1,N_2}$ by $1/\sqrt{N_1N_2}$.
Out of such similar choices, the one that we make here is the physically meaningful one, for it yields the physically correct weights in the Hartree equations (as suggested by experiments).
A justification for our choice is provided in \cite[Section~4]{M-Olg-2016_2mixtureMF}.

A natural example of state which models a two-component condensate is a state $\psi_{N_1,N_2} \in \mathcal{H}$ of the form
\[
  \psi_{N_1,N_2}(x_1,\dots,x_{N_1},y_1,\dots,y_{N_2}) = \prod_{j=1}^{N_1} u(x_j) \otimes \prod_{k=1}^{N_2} v(y_k)
\]
for some $u, v \in H^1(\R^3)$ with $\| u \|_{L^2} = 1$ and $\|v\|_{L^2} = 1$.
If $\psi_{N_1,N_2}$ represents a condensate at $t = 0$, its time-evolution $\psi_{N_1,N_2,t}$ is determined by the Schr\"{o}dinger equation
\begin{equation}
  \label{schroedinger}
  i \partial_t \psi_{N_1,N_2,t} = H_{N_1,N_2} \psi_{N_1,N_2,t} \quad \text{with} \quad \psi_{N_1,N_2,0} = \psi_{N_1,N_2}.
\end{equation}
In order to consider systems with variable number of particles, we will embed states of the form $\psi_{N_1,N_2}$ into a product of Fock spaces by considering a coherent superposition of such states.
More precisely, we will consider a tensor product of coherent states as initial state (more on this in Section \ref{s:fock}).

The factorization of the initial state is not preserved by the time-evolution.
However, in the mean-field regime, similarly to what happens for one-component condensates \cite{RS-2007}, we expect that the solution $\psi_{N_1,N_2,t}$ is approximately factorized (in a sense describe below) when $N_1 \to \infty$ and $N_2 \to \infty$.
Schematically, we expect that
\[
  \psi_{N_1,N_2,t}(x_1,\dots,x_{N_1},y_1,\dots,y_{N_2}) \simeq \prod_{j=1}^{N_1} u_t(x_j) \otimes \prod_{k=1}^{N_2} v_t(y_k)
\]
where $u_t$ and $v_t$ are the solutions of the system of equations of Hartree type
\begin{equation}
  \label{sysHartree}
  \begin{aligned}
    i \partial_t u_t & = -\Delta u_t + (V_1 * |u_t|^2) u_t + c_2 (V_{12} * |v_t|^2) u_t \\
    i \partial_t v_t & = -\Delta v_t + (V_2 * |v_t|^2) v_t + c_1 (V_{12} * |u_t|^2) v_t
  \end{aligned}
\end{equation}
with $u_0 = u$ e $v_0 = v$.
We will specify later the constants $c_1$ and $c_2$.

We now explain in which sense the approximate factorization holds true.
When the mixture of bosons is in the state $\psi_{N_1,N_2,t}$, condensation in each component can be inferred by using the density operator
\begin{equation}\label{eq:def_double_partial_trace-ABSTRACT}
  \gamma_{{N_1,N_2,t}}^{(1,1)} = \mathrm{Tr}_{N_1-1, N_2-1} \, |\psi_{N_1,N_2,t}\rangle\langle \psi_{N_1,N_2,t}|,
\end{equation}
where $|\psi_{N_1,N_2,t}\rangle\langle \psi_{N_1,N_2,t}|$ denotes the orthogonal projection onto the state $\psi_{N_1,N_2,t}$ and $\mathrm{Tr}_{N_1-1,N_2-1}$ denotes the trace over $N_1-1$ variables corresponding to the first factor of $\mathcal{H}$ and $N_2-1$ variables corresponding to the second factor.
The operator $\gamma_{{N_1,N_2,t}}^{(1,1)}$ is a non-negative trace class operator on $L^2(\mathbb{R}^3)\otimes L^2(\mathbb{R}^3)$ with integral kernel
\begin{equation}\label{eq:def_double_partial_trace-KERNEL}
\begin{split}
  \gamma_{{N_1,N_2,t}}^{(1,1)}(x,x';y,y')\;=\;&\int_{\mathbb{R}^{3(N_1-1)}} \int_{\mathbb{R}^{3(N_2-1)}} \ud x_2\cdots\ud x_{N_1}\ud y_2\cdots\ud y_{N_2} \\
& \qquad\times \psi_{N_1,N_2,t}(x,x_2,\dots,x_{N_1};y,y_2,\dots,y_{N_2}) \\
& \qquad\times \overline{\psi_{N_1,N_2,t}}(x',x_2,\dots,x_{N_1};y',y_2,\dots,y_{N_2})\,.
\end{split}
\end{equation}
If $\| \psi_{N_1,N_2,t} \| = 1$, we have $\gamma_{{N_1,N_2,t}}^{(1,1)} = 1$.

For $t > 0$, the operator $\gamma_{N_1,N_2,t}^{(1,1)}$ is not rank-one and is not factorised as a tensor product of two density operators. The special case of (complete) two-component BEC with condensate functions $u_t$ and $v_t$ corresponds to the situation when
\begin{equation}\label{eq:def_100BEC-2component}
\lim_{\substack{N_1\to\infty \\ N_2\to\infty}}\gamma_{N_1,N_2,t}^{(1,1)}\;=\;|u_t\otimes v_t\rangle\langle u_t\otimes v_t|,
\end{equation}
where convergence occurs with respect to the trace norm.
The limit expresses the fact that the actual many-body state has the same occupation numbers of the pure tensor product $u_t^{\otimes N_1}\otimes v_t^{\otimes N_2}$. In fact, $\gamma_{N_1,N_2,t}^{(1,1)}$ has non-negative real eigenvalues that sum up to $1$ and that are naturally interpreted as the fraction of the particles occupying the corresponding eigenstates. Thus, when \eqref{eq:def_100BEC-2component} occurs, it means that there is macroscopic occupation of bosons of the first species in the one-body state $u_t$ and a macroscopic occupation of bosons of the second species in the one-body state $v_t$.
The condition \eqref{eq:def_100BEC-2component} is the precise meaning of the approximate factorization mentioned above.

In the limit, the vanishing of $\gamma_{N_1,N_2,t}^{(1,1)}-|u_t\otimes v_t\rangle\langle u_t\otimes v_t|$ (with respect to the trace norm) is much weaker than the vanishing of $\|\psi_{N_1,N_2,t}-u_t^{\otimes N_1}\otimes v_t^{\otimes N_2}\|_{\mathcal{H}}$. In fact, unless the system is non-interacting, even in the regime of condensation there is an additional inter-particle correlation structure in the typical $\psi_{N_1,N_2,t}$ that can not be described using only $\gamma_{N_1,N_2,t}^{(1,1)}$ (this subject, for one-component condensates, has been investigated in \cite{Lewin-Nam-Serfaty-Solovej-2012_Bogolubov_Spectrum_Interacting_Bose,Lewin-Nam-Schlein-2013_Fluctuations_around_Hartree,Nam-Seiringer-2014_Collective_Excit_Bose}).

In this work, we are interested in the time-evolution of a two-component mixture condensate once the gas is prepared in a state of complete condensation in each component. As in experiments, after the initial state is prepared, we imagine that the system evolves under the sole effect of the intra- and inter-species interactions. Under suitable conditions on the density and the interactions of the system, the two-component condensation persists at later times and the system is condensed onto two one-body orbitals that are solutions of the system of equations \eqref{sysHartree}.

The study of the effective dynamics of Bose-Einstein condensates in various regimes is a topic of current interest \cite{M-Olg-2016_2mixtureMF,AO-GPmixture-2016volume,Anap-Hott-Hundertmark-2017,MO-pseudospinors-2017}.
In \cite{M-Olg-2016_2mixtureMF}, a paper by one of us and A.~Olgiati, effective evolution equations were derived rigorously from the many-body Schr\"{o}dinger equation by using a method of ``counting'' the number of particles in the many-body state that occupy the one-body orbitals, at time $t$. This was possible by adapting Pickl's counting method which was used in the study of one-component condensates \cite{kp-2009-cmp2010,Pickl-JSP-2010,Pickl-LMP-2011,Pickl-RMP-2015}.
In \cite{Anap-Hott-Hundertmark-2017}, a similar result was obtained by Anapolitanos, Hott and Hundertmark for systems of particles interacting through the Yukawa potential.

Our paper is related to the analysis done in \cite{M-Olg-2016_2mixtureMF} (see also \cite{MO-pseudospinors-2017}) and \cite{Anap-Hott-Hundertmark-2017}---While our results are similar to the ones in these articles, we formulate our model in a different setting, and we use different methods.
We derive a system of effective equations of Hartree type by using methods in Fock space developed by Hepp \cite{hepp-1974}, Ginibre and Velo \cite{Ginibre-Velo}, Rodnianski and Schlein \cite{RS-2007} and others (see also \cite{Benedikter-DeOliveira-Schlein_QuantitativeGP-2012_CPAM2015}).
A review of these methods for one-component Bose gases can be found in \cite{Benedikter-Porta-Schlein-2015}.
Our results hold for coherent states as initial data (Theorem \ref{t:coherent}) and they allow the Coulomb interaction.
In short, the methods in Fock space are based on the idea of controlling fluctuations of the dynamics with respect to the leading (mean-field) dynamics.
They have been used to study dynamical properties of one-component condensates \cite{chen-lee-2010-JMP2011,Chen-Lee-Schlein-2011,Benedikter-DeOliveira-Schlein_QuantitativeGP-2012_CPAM2015,Lewin-Nam-Schlein-2013_Fluctuations_around_Hartree,Boccato-Cenatiempo-Schlein-2015_AHP2017_fluctuations,Brennecke-Schlein-GP2017}, and also to study fermionic systems \cite{BPS, PP}.
As another example of application of Fock space methods, we mention the recent work of Chen and Soffer \cite{ChenSoffer}.

In order to state the main result of this paper (Theorem \ref{t:coherent}), we need to describe our model in Fock space.
This is done in Section \ref{s:fock}.
The proof of Theorem \ref{t:coherent} appears in Section \ref{s:proofofcoherent}.
In Sections \ref{s:proofofp:part1} and \ref{s:proofofp:part2}, we prove some propositions that we used to prove Theorem \ref{t:coherent}.

\section{The model in Fock space and the main result}
\label{s:fock}

We want to represent our model for a two-component condensate using a tensor product of Fock spaces.
To do this, we will proceed as follows.
First, we present the basic definitions regarding the Fock space and some operators defined on it.
Then, we describe the tensor product of Fock spaces and the tensor products of some operators.
This gives a mathematical model for the two-component condensate.

\textbf{Spaces and basic operators.}
The Fock space for bosons over $L^2(\R^3)$, denoted by $\mathcal{F}$, is the direct sum
\[
  \mathcal{F} = \bigoplus_{n \ge 0} L^2_s(\R^{3n}).
\]
Here, we are using the convention that $L^2_s(\R^0) = \C$.
Given $\psi \in \mathcal{F}$, we have $\psi = \psi^{(0)} \oplus \psi^{(1)} \oplus \psi^{(2)} \oplus \cdots$ with $\psi^{(n)} \in L^2_s(\R^{3n})$ for $n \ge 0$.
For $\psi, \phi \in \mathcal{F}$, the inner product is given by
\[
  \langle \psi, \phi \rangle = \sum_{n \ge 0} \langle \psi^{(n)}, \phi^{(n)} \rangle.
\]
The vector space $\mathcal{F}$ with this inner product is a Hilbert space.
The induced norm on this space is denoted by $\| \cdot \|$.

We now mention some special states in $\mathcal{F}$ which are relevant in our work.
First, the vacuum state $1 \oplus 0 \oplus 0 \oplus \cdots$, denoted by $\Omega$.
Secondly, states with exactly $n$ particles, denoted by $\psi_{\underline{n}}$.
Thus $\psi_{\underline{n}} = 0 \oplus 0 \oplus \cdots \oplus \psi_{\underline{n}}^{(n)} \oplus 0 \oplus \cdots$.

An operator that plays an important role in our work is the number of particles operator, denoted by $\mathcal{N}$.
This is the self-adjoint operator on $\mathcal{F}$ defined by
\[
  \mathcal{N} \psi = 0 \cdot \psi^{(0)} \oplus 1 \cdot \psi^{(1)} \oplus 2 \cdot \psi^{(2)} \oplus \cdots \oplus n \cdot \psi^{(n)} \oplus \cdots
\]
for any $\psi \in \mathcal{F}$ such that $\sum_{n \ge 0} n^2 \| \psi^{(n)} \|^2 < \infty$.
We observe that a state $\psi_{\underline{n}}$ with exactly $n$ particles is an eigenvector of $\mathcal{N}$ with eigenvalue $n$.

Our next step is defining creation and annihilation operators on $\mathcal{F}$.
For $f$ in $L^2(\R^3)$, the creation operator $a^*(f)$ is defined (as the closure of)
\[
  (a^*(f) \psi)^{(n)}(x_1, \dots, x_n) = \frac{1}{\sqrt{n}} \sum_{j=1}^n f(x_j) \psi^{(n-1)}(x_1, \dots, x_{j-1}, x_{j+1}, \dots, x_n)
\]
for $n = 1, 2, \dots$ and $(a^*(f) \psi)^{(0)} = 0$.
Here $\psi = \psi^{(0)} \oplus \psi^{(1)} \oplus \psi^{(2)} \oplus \cdots$ and $a(f) \psi = (a(f) \psi)^{(0)} \oplus (a(f) \psi)^{(1)} \oplus \cdots$.
Similarly, the annihilation operator is defined (as the closure of)
\[
  (a(f) \psi)^{(n)}(x_1, \dots, x_n) = \sqrt{n+1} \int dx \, \overline{f(x)} \psi^{(n+1)}(x,x_1, \dots, x_n)
\]
for $n = 0, 1, \dots$ and $a(f) \Omega = 0$.
The operator $a^*(f)$ is the adjoint of $a(f)$.
Both operators are unbounded, densely defined and closed.
For $f$ and $g$ in $L^2(\R^3)$, they obey the canonical commutation relations
\[
  [a(f), a^*(g)] = \langle f, g \rangle \qquad \text{and} \qquad [a(f), a(g)] = [a^*(f), a^*(g)] = 0.
\]
Related to $a^*(f)$ and $a(f)$, it is useful to consider the self-adjoint operator
\[
  \phi(f) = a^*(f) + a(f).
\]

Creation and annihilation operators can be represented using the operator valued distributions $a^*_x$ and $a_x$ for $x \in \R^3$.
These are the distributions such that
\[
  a^*(f) = \int dx \, f(x) a_x^* \qquad \text{and} \qquad a(f) = \int dx \, \overline{f(x)} a_x.
\]
Consequently, for $x, y \in \R^3$, we have
\[
  [a_x, a_y^*] = \delta(x-y) \qquad \text{and} \qquad [a_x, a_y] = [a_x^*, a_y^*] = 0.
\]

Using the operator valued distributions, we can represent the number of particles operator as
\[
  \mathcal{N} = \int dx \, a_x^* a_x.
\]

We now consider two copies of the bosonic Fock space $\mathcal{F}$.
In our model, the state space for the two-component condensate is $\mathcal{F} \otimes \mathcal{F}$.
Here, we are using the standard tensor product of two Hilbert spaces (as described in \cite{ReedSim1}, for example).

Having described the state space, let us define operators on it.
We will often use the standard construction of tensor product of operators, as described in \cite{ReedSim1}, for example.

For $f \in L^2(\R^3)$, we define creation and annihilation operators on each factor of $\mathcal{F} \otimes \mathcal{F}$ by
\begin{alignat*}{3}
  b^*(f) & = a^*(f) \otimes I, & \qquad \qquad c^*(f) & = I \otimes a^*(f), \\
  b(f) & = a(f) \otimes I, & \qquad \qquad c(f) & = I \otimes a(f).
\end{alignat*}
The corresponding operator valued distributions are given by
\[
  b_x^* = a_x^* \otimes I, \qquad c_x^* = I \otimes a_x^*, \qquad b_x = a_x \otimes I, \qquad c_x = I \otimes a_x.
\]
There are several commutation relations for the operators $b^*(f)$, $c^*(f)$, $b(f)$ and $c(f)$.
The only commutators that are not equal to zero are
\[
  [b(f), b^*(g)] = \langle f, g \rangle \qquad \text{and} \qquad [c(f), c^*(g)] = \langle f, g \rangle.
\]
Consequently
\[
  [b_x, b_y^*] = \delta(x-y) \qquad \text{and} \qquad [c_x, c_y^*] = \delta(x-y).
\]

Using the operator-valued distributions, we can represent $\mathcal{N} \otimes I$ and $I \otimes \mathcal{N}$ as
\[
  \mathcal{N} \otimes I = \int dx \, b_x^* b_x \qquad \text{and} \qquad I \otimes \mathcal{N} = \int dy \, c_y^* c_y.
\]

The creation and annihilation operators are bounded with respect to $I \otimes \mathcal{N}^{1/2}$, $\mathcal{N}^{1/2} \otimes I$, etc.
The precise statement is given in the following lemma.
The proof of the lemma follows from the corresponding well-known result for creation and annihilation operators on each factor of $\mathcal{F} \otimes \mathcal{F}$ (see \cite{RS-2007} for a proof).

\begin{lemma}
  \label{l:relbN}
  Let $f, g \in L^2(\R^3)$.
  For any $\psi \in \mathcal{F} \otimes \mathcal{F}$, we have
  \begin{alignat*}{3}
    \| b^*(f) \psi \| & \le \| f \| \, \| \big( (\mathcal{N} + 1)^{1/2} \otimes I \big) \psi \|, & \qquad \| c^*(g) \psi \| & \le \| g \| \, \| \big( I \otimes (\mathcal{N} + 1)^{1/2} \big) \psi \|, \\
    \| b(f) \psi \| & \le \| f \| \, \| \big( \mathcal{N}^{1/2} \otimes I \big) \psi \|, & \qquad \| c(g) \psi \| & \le \| g \| \, \| \big( I \otimes \mathcal{N}^{1/2} \big) \psi \|.
  \end{alignat*}
\end{lemma}

\textbf{The Hamiltonian.}
We now define the Hamiltonian $\mathcal{H}_{N_1,N_2}$ acting on $\mathcal{F} \otimes \mathcal{F}$.
First, we observe that the set of all finite sums $\sum_k \psi_{\underline{k}}$ is dense in~$\mathcal{F}$.
Moreover, the set of all finite linear combinations of products $\psi_{\underline{n_1}} \otimes \phi_{\underline{n_2}}$, denoted by $\mathcal{D}$, is dense in $\mathcal{F} \otimes \mathcal{F}$.
We define the Hamiltonian $\mathcal{H}_{N_1,N_2}$ acting on products $\psi_{\underline{n_1}} \otimes \phi_{\underline{n_2}}$ by $\mathcal{H}_{N_1,N_2} \, \psi_{\underline{n_1}} \otimes \phi_{\underline{n_2}} = \mathcal{H}_{N_1,N_2}^{(n_1,n_2)} \, \psi_{\underline{n_1}} \otimes \phi_{\underline{n_2}}$ where
\[
  \mathcal{H}_{N_1,N_2}^{(n_1,n_2)} = h_{N_1}^{(n_1)} \otimes I + I \otimes h_{N_2}^{(n_2)} + T_{N_1,N_2}^{(n_1,n_2)}
\]
with
\[
  h_{N_p}^{(n_p)} = \sum_{j=1}^{n_p} -\Delta_{x_j} + \frac{1}{N_p} \sum_{1 \le j < k \le n_p} V_p(x_j-x_k)
\]
and
\[
  T_{N_1,N_2}^{(n_1,n_2)} = \frac{1}{N_1+N_2} \sum_{j=1}^{n_1} \sum_{k=1}^{n_2} V_{12}(x_j-y_k).
\]
We then extend $\mathcal{H}_{N_1,N_2}$ to $\mathcal{D}$ by linearity.
With the hypothesis on $V_1$, $V_2$ and $V_{12}$ in Theorem \ref{t:coherent}, the Hamiltonian $\mathcal{H}_{N_1,N_2}$ on $\mathcal{D}$ gives rise to a self-adjoint operator (which we denote by the same symbol $\mathcal{H}_{N_1,N_2}$) \cite[Theorem X.23]{ReedSim2}.
In particular, the initial value problem $i \partial_t \psi_t = \mathcal{H}_{N_1,N_2} \psi_t$ with $\psi_t|_{t=0} = \psi_0$ is well-posed.

Using the operator-valued distributions $b_x$ and $c_x$, the Hamiltonian $\mathcal{H}_{N_1,N_2}$ can be written as
\[
  \mathcal{H}_{N_1,N_2} = \mathcal{H}_{N_1} + \mathcal{H}_{N_2} + \mathcal{T}_{N_1,N_2}
\]
where
\[
  \mathcal{H}_{N_1} = \int dx \, \nabla_x b_x^* \nabla_x b_x + \frac{1}{N_1} \int \int dx dz \, V_1(x-z) b_x^* b_z^* b_z b_x,
\]
\[
  \mathcal{H}_{N_2} = \int dy \, \nabla_y c_y^* \nabla_y c_y + \frac{1}{N_2} \int \int dy dz \, V_2(y-z) c_y^* c_z^* c_z c_y,
\]
and
\[
  \mathcal{T}_{N_1,N_2} = \frac{1}{N_1+N_2} \int \int dx dy \, V_{12}(x-y) b_x^* c_y^* c_y b_x.
\]

The Hamiltonian $\mathcal{H}_{N_1,N_2}$ conserves the number of particles in each factor of $\mathcal{F} \otimes \mathcal{F}$.
In fact, it is simple to verify that, for $j = 1, 2$,
\[
  [\mathcal{H}_{N_j}, \mathcal{N} \otimes I] = [\mathcal{H}_{N_j}, I \otimes \mathcal{N}] = [\mathcal{T}_{N_1,N_2}, \mathcal{N} \otimes I] = [\mathcal{T}_{N_1,N_2}, I \otimes \mathcal{N}] = 0.
\]
Furthermore, for fixed $N_1$ and $N_2$, the subspace
\[
  \mathcal{S}_{N_1,N_2} = \text{span} \{ \psi_{\underline{N_1}} \otimes \phi_{\underline{N_2}} \, | \, \psi_{\underline{N_1}} \in L^2_s(\R^{3N_1}), \ \phi_{\underline{N_2}} \in L^2_s(\R^{3N_2} \}
\]
is invariant by $\mathcal{H}_{N_1,N_2}$, and the Hamiltonian $\mathcal{H}_{N_1,N_2}$ restricted to $\mathcal{S}_{N_1,N_2}$ is equal to $H_{N_1,N_2}$.
Therefore, for initial data in $\mathcal{S}_{N_1,N_2}$, the time evolution generated by $\mathcal{H}_{N_1,N_2}$ reduces to the time evolution generated by $H_{N_1,N_2}$.

\textbf{The reduced density operator.}
For $\psi \in \mathcal{F} \otimes \mathcal{F}$, we define the reduced density operator $\gamma_\psi^{(1,1)}$ as the operator on $L^2(\R^3) \otimes L^2(\R^3)$ determined by the kernel
\begin{equation}
  \label{intKer}
  \gamma_\psi^{(1,1)}(x, y; x', y') = \frac{1}{\langle \psi, \mathcal{N} \otimes \mathcal{N} \psi \rangle} \langle \psi, b_{x'}^* c_{y'}^* b_x c_y \psi \rangle.
\end{equation}
We observe that $\text{Tr}_{1,1} \gamma_{\psi}^{(1,1)} = 1$.
If $\psi$ is in the subspace $\mathcal{S}_{N_1,N_2}$ of fixed number of particles, the above definition reduces to the definition of $\gamma_{N_1,N_2,t}^{(1,1)}$ given in the introduction.

\textbf{Coherent states.}
For $f \in L^2(\R^3)$, the Weyl operator on $\mathcal{F}$, denoted $W(f)$, is defined by
\[
  W(f) = \exp(a^*(f) - a(f)).
\]
The state $W(f) \Omega$ is called a coherent state.
We have
\begin{equation}
  \label{WOmega}
  W(f) \Omega = e^{-\|f\|^2/2} \sum_{n \ge 0} \frac{a^*(f)^n}{n!} \Omega = e^{-\|f\|^2/2} \sum_{n \ge 0} \frac{1}{\sqrt{n!}} f^{\otimes n}.
\end{equation}
This expression is obtained by using the identities
\[
  \exp(a^*(f) - a(f)) = e^{-\| f \|^2/2} \exp(a^*(f)) \exp(a(f))
\]
and $a(f) \Omega = 0$.
The first identity is obtained using $[a(f), a^*(f)] = \| f \|^2$.

For $f, g \in L^2(\R^3)$, we define
\[
  \mathcal{W}(f,g) = W(f) \otimes W(g),
\]
which is an operator on $\mathcal{F} \otimes \mathcal{F}$.
We set $\omega = \Omega \otimes \Omega$ (which we also call a vacuum state).
Thus $\mathcal{W}(f,g) \omega$ is a tensor product of coherent states (which we also call a coherent state).

In the following lemma, we have some important properties of the operator $\mathcal{W}(f,g)$ and the coherent state $\mathcal{W}(f,g)\omega$.
These properties follow easily from the corresponding well-known properties of Weyl operators (see \cite{RS-2007}, for example).

\begin{lemma}
  \label{l:weyl}
  Let $f, g \in L^2(\R^3)$ and $\omega = \Omega \otimes \Omega$.
  \begin{enumerate}[\rm (a)]
    \item
      The operator $\mathcal{W}(f,g)$ is unitary and
      \[
        \mathcal{W}(f,g)^* = \mathcal{W}(f,g)^{-1} = \mathcal{W}(-f,-g).
      \]
    \item
      We have
      \begin{align*}
        \mathcal{W}(f,g)^* b_x \mathcal{W}(f,g) & = b_x + f(x), & \qquad \mathcal{W}(f,g)^* b_x^* \mathcal{W}(f,g) & = b_x^* + \overline{f(x)}, \\
        \mathcal{W}(f,g)^* c_y \mathcal{W}(f,g) & = c_y + g(y), & \qquad \mathcal{W}(f,g)^* c_y^* \mathcal{W}(f,g) & = c_y^* + \overline{g(y)}.
      \end{align*}
    \item
      For $p, q \in \setl{0, 1}$, we have
      \[
        \la \mathcal{W}(f,g) \omega, (\mathcal{N}^p \otimes \mathcal{N}^q) \mathcal{W}(f,g) \omega \ra = \norm{f}^{2p} \norm{g}^{2q}.
      \]
  \end{enumerate}
\end{lemma}

\textbf{Schr\"odinger dynamics.}
For $u, v \in H^1(\R^3)$ with $\norm{u}_{L^2} = 1$ and $\norm{v}_{L^2} = 1$, we set
\[
  \psi_{N_1,N_2} = \mathcal{W}(\sqrt{N_1} u, \sqrt{N_2} v) \omega
\]
and denote by $t \mapsto \psi_t = e^{-i t \mathcal{H}_{N_1,N_2}} \psi_{N_1,N_2}$ the solution to the Schr\"odinger equation $i \partial_t \psi_t = \mathcal{H}_{N_1,N_2} \psi_t$ with initial condition $\psi_0 = \psi_{N_1,N_2}$.
We will study the family of solutions $\{ \psi_t \}_{N_1,N_2}$ as $N_1$ and $N_2$ go to infinity.

We are ready to state our main result:

\begin{theorem}
  \label{t:coherent}
  Suppose that the functions $V = V_1$, $V = V_2$ and $V = V_{12}$ satisfy the operator inequality
  \[
    V^2 \le K (1 - \Delta)
  \]
  on $L^2(\R^3)$ for some constant $K > 0$.
  Suppose that $(N_1)$ and $(N_2)$ are sequences of positive integers and $c_1$ and $c_2$ are real numbers which obey $N_1 \to \infty$ and $N_2 \to \infty$ with
  \begin{equation}
    \label{condseq}
    \left| \frac{N_1}{N_1 + N_2} - c_1 \right| \le \frac{D}{N_2} \qquad \text{and} \qquad \left| \frac{N_2}{N_1 + N_2} - c_2 \right| \le \frac{D}{N_1},
  \end{equation}
  for some constant $D > 0$, where $c_1 \geq 0$ and $c_2 \geq 0$ with $c_1 + c_2 = 1$.
  Let $\gamma_{t}^{(1,1)}$ be the reduced density operator associated to the solution
  \[
    \psi_t = e^{-i t \mathcal{H}_{N_1,N_2}} \mathcal{W}(\sqrt{N_1} u, \sqrt{N_2} v) \omega
  \]
  to the Schr\"{o}dinger equation.
  Then
  \begin{equation}
    \label{ineqcoherent}
    \text{\rm Tr} \, \Big| \gamma_{t}^{(1,1)} - |u_t \otimes v_t \rangle \langle u_t \otimes v_t| \, \Big| \le C e^{\gamma t} \left[ \frac{1}{\sqrt{N_1}} + \frac{1}{\sqrt{N_2}} \right]
  \end{equation}
  for all $t \ge 0$, where $C$ and $\gamma$ are positive constants that do not depend on $N_1$ and $N_2$, and $u_t$ and $v_t$ are the solutions of the system of equations \eqref{sysHartree} with initial datum $u$ and $v$.
\end{theorem}

We make the following observations about Theorem \ref{t:coherent}:
\begin{enumerate}
  \item
    The hypothesis on the interaction potentials $V_1$, $V_2$ and $V_{12}$ allow the Coulomb potential.
    In addition, the theorem holds for both attractive and repulsive potentials.
  \item
    Under the hypothesis on the interaction potentials and the initial datum $u$ and $v$, the system of equations \eqref{sysHartree} of Hartree type is globally well-posed in time.
    This follows from conservation of mass and energy.
\end{enumerate}

If $\psi_t$ was equal to the coherent state $\mathcal{W}(\sqrt{N_1} u_t, \sqrt{N_2} v_t) \omega$, the left hand side of \eqref{ineqcoherent} would vanish identically (and Theorem \ref{t:coherent} would be true with $C = 0$).
This observation suggests a strategy to prove Theorem \ref{t:coherent}.
Namely, the idea is to show that $\psi_t \approx \mathcal{W}(\sqrt{N_1} u_t, \sqrt{N_2} v_t) \omega$ in a certain sense.
In order to do this, we make a time-dependent change of variable which transforms the Schr\"odinger equation into a non-autonomous equation, called the fluctuation equation, whose initial datum is the vacuum state.
Again, if the vacuum state was a stationary solution to this equation, the left hand side of \eqref{ineqcoherent} would vanish identically.
The vacuum state is not a stationary solution, but the actual solution to the fluctuation equation has a property that suffices to obtain Theorem \ref{t:coherent}.
Namely, the expected value of the number of particles for the solution does not grow too fast as $N_1$ and $N_2$ go to infinity.
Below we describe how to implement this strategy to prove Theorem \ref{t:coherent}.
As we mentioned earlier, we will use the methods developed in \cite{RS-2007} (see also \cite{Benedikter-DeOliveira-Schlein_QuantitativeGP-2012_CPAM2015}).

To prove Theorem \ref{t:coherent}, we will proceed as follows:
We will consider a tensor product of coherent states as initial data.
The Hartree dynamics emerges as the main component of the evolution of this initial data (in the mean field limit).
The problem reduces to study fluctuations around this main component.
(These fluctuations are described by a two-parameter unitary group $U_{N_1,N_2}(s,t)$ called fluctuation dynamics.)
We will prove that the number of fluctuations is controlled, in certain sense, for large $N_1$ and $N_2$.
This implies the statement in Theorem \ref{t:coherent}.

We now turn to the proof of the theorem.

\textbf{Fluctuation dynamics.}
We define $\alpha(t) = \int_0^t F_{N_1,N_2}(s) \, ds$, where $F_{N_1,N_2}$ is a real-valued function that we will choose later.
Let $u_t$ and $v_t$ be the solutions to the Hartree system \eqref{sysHartree}.
For $t, s \in \R$, we set
\[
  U(t,s) = e^{i(\alpha(t) - \alpha(s))} \mathcal{W}(\sqrt{N_1} u_t, \sqrt{N_2} v_t)^* e^{-i(t-s) \mathcal{H}_{N_1,N_2}} \mathcal{W}(\sqrt{N_1} u_s, \sqrt{N_2} v_s).
\]
We refer to the operator $U(t,s)$ as the fluctuation dynamics.
We abbreviate
\[
  \mathcal{W}_t = \mathcal{W}(\sqrt{N_1} u_t, \sqrt{N_2} v_t) \qquad \text{and} \qquad U_{t,s} = U(t,s).
\]
Thus, we may write $U_{t,s} = e^{i (\alpha(t) - \alpha(s))} \mathcal{W}_t e^{-i(t-s) \mathcal{H}_{N_1,N_2}} \mathcal{W}_s$.
We also define
\[
  \omega_t = e^{-i \alpha(t)} U(t,0) \omega.
\]
Using the above definitions, it is simple to verify that
\[
  \psi_t = \mathcal{W}_t \, \omega_t.
\]

For each $s \in \R$, the fluctuation dynamics satisfies the equation
\[
  i \partial_t U_{t,s} = \mathcal{L}(t) U_{t,s} \qquad \text{with} \qquad U_{s,s} = I
\]
where
\[
  \mathcal{L}(t) = (i \partial_t \mathcal{W}_t^*) \mathcal{W}_t + \mathcal{W}_t^* \mathcal{H}_{N_1,N_2} \mathcal{W}_t - F_{N_1,N_2}(t).
\]
We next calculate the first two terms in this expression.

To calculate $(i \partial_t \mathcal{W}_t^*) \mathcal{W}_t$, we use the identity
\[
  (\partial_t e^{-A(t)}) e^{A(t)} = - \int_0^1 d\lambda \, e^{-A(t)\lambda} \dot{A}(t) e^{A(t) \lambda}.
\]
Here $A(t)$ is an operator.
(This identity is obtained using the definition of time derivative and the formula $e^{-A}Be^A = B - \int d\lambda e^{-\lambda A} [A,B] e^{\lambda A}$.)
We find that
\[
  \begin{split}
    (i \partial_t \mathcal{W}_t^*) \mathcal{W}_t & = \sqrt{N}_1 \phi( -i \partial_t u_t) \otimes I + \sqrt{N}_2 I \otimes \phi(-i \partial_t v_t) \\
    & \quad + N_1 \Re \langle u_t, i \partial_t u_t \rangle + N_2 \Re \langle v_t, i \partial_t v_t \rangle.
  \end{split}
\]

To calculate $\mathcal{W}_t^* \mathcal{H}_{N_1,N_2} \mathcal{W}_t$, we use Lemma \ref{l:weyl}(b).
We obtain
\[
  \begin{split}
    & \mathcal{W}_t^* \mathcal{H}_{N_1,N_2} \mathcal{W}_t \\
    & = \mathcal{H}_{N_1,N_2} + C_{N_1,N_2}(t) + Q_{N_1,N_2}(t) \\
    & \quad + \sqrt{N_1} \phi \left( -\Delta u_t + (V_1 * |u_t|^2) u_t + \frac{N_2}{N_1+N_2} (V_{12} * |v_t|^2) u_t \right) \otimes I \\
    & \quad + \sqrt{N_2} I \otimes \phi \left( -\Delta v_t + (V_2 * |v_t|^2) v_t + \frac{N_1}{N_1+N_2} (V_{12} * |u_t|^2) v_t \right) \\
    & \quad + G_{N_1,N_2}(t)
  \end{split}
\]
where
\[
  \begin{split}
    C_{N_1,N_2}(t) & = \frac{1}{\sqrt{N_1}} \int dx dz \, V_1(x-z) b_x^* (u_t(z) b_z^* + \overline{u_t}(z) b_z) b_x \\
    & \quad + \frac{1}{\sqrt{N_2}} \int dy dz \, V_2(y-z) c_y^* (v_t(z) c_z^* + \overline{v_t}(z) c_z) c_y \\
    & \quad + \frac{\sqrt{N_1}}{N_1 + N_2} \int dx dy \, V_{12}(x-y) c_y^* ( u_t(x) b_x^* + \overline{u_t}(x) b_x ) c_y \\
    & \quad + \frac{\sqrt{N_2}}{N_1 + N_2} \int dx dy \, V_{12}(x-y) b_x^* ( v_t(y) c_y^* + \overline{v_t}(y) c_y ) b_x,
  \end{split}
\]
\[
  \begin{split}
    & Q_{N_1,N_2}(t) \\
    & = \int dx \, (V_1 * |u_t|^2)(x) b_x^* b_x + \frac{N_2}{N_1+N_2} \int dx \, (V_{12} * |v_t|^2)(x) b_x^* b_x \\
    & \quad + \int dy \, (V_2 * |v_t|^2)(y) c_y^* c_y + \frac{N_2}{N_1+N_2} \int dy \, (V_{12} * |u_t|^2)(y) c_y^* c_y \\
    & \quad + \int dx dz \, V_1(x-z) \overline{u_t}(z) u_t(x) b_x^* b_z + \int dy dz \, V_2(y-z) \overline{v_t}(z) v_t(y) c_y^* c_z \\
    & \quad + \frac{1}{2} \int dx dz \, V_1(x-z) ( u_t(z) u_t (x) b_x^* b_z^* + \overline{u_t}(z) \overline{u_t}(x) b_x b_z ) \\
    & \quad + \frac{1}{2} \int dy dz \, V_2(y-z) ( v_t(z) v_t(y) c_y^* c_z^* + \overline{v_t}(z) \overline{v_t}(y) c_y c_x ) \\
    & \quad + \frac{\sqrt{N_1 N_2}}{N_1 + N_2} \int dx dy \, V_{12}(x-y) ( \overline{v_t}(y) u_t(x) b_x^* c_y + \overline{u_t}(x) v_t(y) c_y^* b_x ) \\
    & \quad + \frac{\sqrt{N_1 N_2}}{N_1 + N_2} \int dx dy \, V_{12}(x-y) ( u_t(x) v_t(y) b_x^* c_y^* + \overline{u_t}(x) \overline{v_t}(y) b_x c_y )
  \end{split}
\]
and
\[
  \begin{split}
    G_{N_1,N_2}(t) & = N_1 \int dx \, |\nabla u_t(x)|^2 + \frac{N_1}{2} \int dx dy \, V_1(x-y) |u_t(x)|^2 |u_t(y)|^2 \\
    & \quad + N_2 \int dy \, |\nabla v_t(y)|^2 + \frac{N_2}{2} \int dx dy \, V_2(x-y) |v_t(x)|^2 |v_y(y)|^2 \\
    & \quad + \frac{N_1 N_2}{N_1+N_2} \int dx dy \, V_{12}(x-y) |u_t(x)|^2 |v_t(y)|^2.
  \end{split}
\]
Here $C_{N_1,N_2}(t)$ and $Q_{N_1,N_2}(t)$ denote operators which are respectively cubic and quadratic in creation or annihilation operators.
The scalar operator $G_{N_1,N_2}(t)$ does not depend on creation or annihilation operators.

By choosing
\begin{equation}
  \label{F}
  F_{N_1,N_2}(t) = N_1 \Re \langle u_t, i \partial_t u_t \rangle + N_2 \Re \langle v_t, i \partial_t v_t \rangle + G_{N_1,N_2}(t),
\end{equation}
we obtain
\[
  \begin{split}
    & \mathcal{L}(t) \\
    & = \mathcal{H}_{N_1,N_2} + C_{N_1,N_2}(t) + Q_{N_1,N_2}(t) \\
    & \quad + \sqrt{N_1} \phi \left( -i \partial_t u_t -\Delta u_t + (V_1 * |u_t|^2) u_t + \frac{N_2}{N_1+N_2} (V_{12} * |v_t|^2) u_t \right) \otimes I \\
    & \quad + \sqrt{N_2} I \otimes \phi \left( -i \partial_t v_t -\Delta v_t + (V_2 * |v_t|^2) v_t + \frac{N_1}{N_1+N_2} (V_{12} * |u_t|^2) v_t \right).
  \end{split}
\]
Since $u_t$ and $v_t$ satisfy \eqref{sysHartree}, we get
\[
  \mathcal{L}(t) = \mathcal{H}_{N_1,N_2} + C_{N_1,N_2}(t) + Q_{N_1,N_2}(t) + R_{N_1,N_2}(t)
\]
where
\[
  \begin{split}
    R_{N_1,N_2}(t) & = \sqrt{N_1} \left( \frac{N_2}{N_1+N_2} - c_2 \right) \phi \big( (V_{12} * |v_t|^2) u_t \big) \otimes I \\
    & \quad + \sqrt{N_2} \left( \frac{N_1}{N_1+N_2} - c_1 \right) I \otimes \phi \big( (V_{12} * |u_t|^2) v_t \big).
  \end{split}
\]

\section{Proof of Theorem \ref{t:coherent}.}
\label{s:proofofcoherent}

For $\mathfrak{h} = L^2(\R^3) \otimes L^2(\R^3)$, let $\mathcal{L}^1(\mathfrak{h})$ be the space of trace class operators on $\mathfrak{h}$ together with the trace norm $\| \cdot \|_1$, and let $\text{Com}(\mathfrak{h})$ be the space of compact operators on $\mathfrak{h}$ together with the operator norm $\| \cdot \|$.
We observe that $\mathcal{L}^1(\mathfrak{h}) \simeq \text{Com}(\mathfrak{h})^*$ with respect to the mapping $\mathcal{L}^1(\mathfrak{h}) \to \text{Com}(\mathfrak{h})^*$ defined by $T \mapsto \text{Tr}(T \, \cdot \,)$.
This mapping is an isometry, an isomorphism, and onto.
Thus, if $| \text{Tr}(T J)| \le C \| J \|$ for all $J \in \text{Com}(\mathfrak{h})$ and $C \!>\! 0$, then $\| T \|_1 = \| \text{Tr}(T \cdot ) \| \le C$.
Therefore, to prove the Theorem \ref{t:coherent}, it suffices to show that
\[
  \Big| \text{Tr} \, J \big( \gamma_{t}^{(1,1)} - |u_t \otimes v_t \rangle \langle u_t \otimes v_t| \, \big) \Big| \le \| J \| C e^{\gamma t} \left[ \frac{1}{\sqrt{N_1}} + \frac{1}{\sqrt{N_2}} \right]
\]
for all $J \in \text{Com}(\mathfrak{h})$.

Since the Hamiltonian conserves the number of particles in each factor of $\mathcal{F} \otimes \mathcal{F}$, by Lemma \ref{l:weyl}(c), we have $\langle \psi_t, \mathcal{N} \otimes \mathcal{N} \psi_t \rangle = N_1 N_2$.
Thus the kernel of $\gamma_{t}^{(1,1)}$ is given by
\begin{equation}
  \label{expgamma}
  \begin{split}
    & \gamma_{t}^{(1,1)}(x, y; x', y') \\
    & = \frac{1}{N_1 N_2} \langle \mathcal{W}_t \omega_t, b_{x'}^* c_{y'}^* b_x c_y \mathcal{W}_t \omega_t \rangle \\
    & = \frac{1}{N_1 N_2} \langle \omega_t, [b_{x'}^* + \sqrt{N_1} \overline{u_t}(x')] [c_{y'}^* + \sqrt{N_2} \overline{v_t}(y')] \\
    & \qquad \qquad \qquad \times [b_x + \sqrt{N_1} u_t(x)] [c_y + \sqrt{N_2} v_t(y)] \omega_t \rangle \\
    & = \overline{u_t}(x') \overline{v_t}(y') u_t(x) v_t(y) + \sum_{j=1}^{15} r_j
  \end{split}
\end{equation}
where
\begin{align*}
  r_1 & = N_1^{-1/2} \, \overline{v_t}(y') u_t(x) v_t(y) \langle \omega_t, b_{x'}^* \omega_t \rangle, \\
  r_2 & = N_2^{-1/2} \, \overline{u_t}(x') u_t(x) v_t(y) \langle \omega_t, c_{y'}^* \omega_t \rangle, \\
  r_3 & = N_1^{-1/2} \, \overline{u_t}(x') \overline{v_t}(y') v_t(y) \langle \omega_t, b_x \omega_t \rangle, \\
  r_4 & = N_2^{-1/2} \, \overline{u_t}(x') \overline{v_t}(y') u_t(x) \langle \omega_t, c_y \omega_t \rangle, \\
  r_5 & = N_1^{-1/2} N_2^{-1/2} \, u_t(x) v_t(y) \langle \omega_t, b_{x'}^* c_{y'}^* \omega_t \rangle, \\
  r_6 & = N_1^{-1} \, \overline{v_t}(y') v_t(y) \langle \omega_t, b_{x'}^* b_x \omega_t \rangle, \\
  r_7 & = N_1^{-1/2} N_2^{-1/2} \, \overline{v_t}(y') u_t(x) \langle \omega_t, b_{x'}^* c_y \omega_t \rangle, \\
  r_8 & = N_1^{-1/2} N_2^{-1/2} \, \overline{u_t}(x') v_t(y) \langle \omega_t, c_{y'}^* b_x \omega_t \rangle, \\
  r_9 & = N_2^{-1} \, \overline{u_t}(x') u_t(x) \langle \omega_t, c_{y'}^* c_y \omega_t \rangle, \\
  r_{10} & = N_1^{-1/2} N_2^{-1/2} \, \overline{u_t}(x') \overline{v_t}(y') \langle \omega_t, b_x c_y \omega_t \rangle, \\
  r_{11} & = N_1 N_2^{-1/2} \, v_t(y) \langle \omega_t, b_{x'}^* c_{y'}^* b_x \omega_t \rangle, \\
  r_{12} & = N_1^{-1/2} N_2 \, u_t(x) \langle \omega_t, b_{x'}^* c_{y'}^* c_y \omega_t \rangle, \\
  r_{13} & = N_1 N_2^{-1/2} \, \overline{v_t}(y') \langle \omega_t, b_{x'}^* b_x c_y \omega_t \rangle, \\
  r_{14} & = N_1^{-1/2} N_2 \, \overline{u_t}(x') \langle \omega_t, c_{y'}^* b_x c_y \omega_t \rangle, \\
  r_{15} & = N_1^{-1} N_2^{-1} \, \langle \omega_t, b_{x'}^* c_{y'}^* b_x c_y \omega_t \rangle.
\end{align*}
We observe that $r_j = r_j(x,y;x',y')$ for $j = 1, \dots, 15$.

Let $J$ be a compact operator on $L^2(\R^3) \otimes L^2(\R^3)$ with integral kernel $J(x',y'; x,y)$ (in the sense of distributions).
We want to estimate the absolute value of
\begin{align*}
  & \text{\rm Tr} \, J \, \big( \gamma_{t}^{(1,1)} - | u_t \otimes v_t \rangle \langle u_t \otimes v_t | \big) \\
  & = \int dx' dy' dx dy \, J(x',y';x,y) \big[ \gamma_{t}^{(1,1)}(x,y;x'y') \!-\! \overline{u_t}(x') \overline{v_t}(y') u_t(x) v_t(y) \big].
\end{align*}
We have the following proposition, which we prove later in Section \ref{s:proofofp:part1}.

\begin{prop}
  \label{p:part1}
  Let $J$ be a compact operator on $L^2(\R^3) \otimes L^2(\R^3)$.
  We have
  \[
    \Big| \text{\rm Tr} J \, \big( \gamma_{t}^{(1,1)} - | u_t \otimes v_t \rangle \langle u_t \otimes v_t | \big) \Big| \le \| J \| \sum_{j=1}^9 p_j
  \]
  where
  \begin{align*}
    p_1 & = 2 N_1^{-1/2} \langle \omega_t, (\mathcal{N} \otimes I) \omega_t \rangle^{1/2}, \\
    p_2 & = 2 N_2^{-1/2} \langle \omega_t, (I \otimes \mathcal{N}) \omega_t \rangle^{1/2}, \\
    p_3 & = 2 N_1^{-1/2} N_2^{-1/2}  \langle \omega_t, (\mathcal{N} \otimes \mathcal{N}) \omega_t \rangle^{1/2}, \\
    p_4 & = 2 N_1^{-1/2} N_2^{-1/2} \langle \omega_t, (\mathcal{N} \otimes I) \omega_t \rangle^{1/2} \langle \omega_t, (I \otimes \mathcal{N}) \omega_t \rangle^{1/2}, \\
    p_5 & = N_1^{-1} \langle \omega_t, (\mathcal{N} \otimes I) \omega_t \rangle, \\
    p_6 & = N_2^{-1} \langle \omega_t, (I \otimes \mathcal{N}) \omega_t \rangle, \\
    p_7 & = 2 N_1^{-1} N_2^{-1/2} \langle \omega_t, (\mathcal{N} \otimes \mathcal{N}) \omega_t \rangle^{1/2} \langle \omega_t, (\mathcal{N} \otimes I) \omega_t \rangle^{1/2}, \\
    p_8 & = 2 N_1^{-1/2} N_2^{-1} \langle \omega_t, (\mathcal{N} \otimes \mathcal{N}) \omega_t \rangle^{1/2} \langle \omega_t, (I \otimes \mathcal{N}) \omega_t \rangle^{1/2}, \\
    p_9 & = N_1^{-1} N_2^{-1} \langle \omega_t, (\mathcal{N} \otimes \mathcal{N}) \omega_t \rangle.
  \end{align*}
  We observe that $p_j = p_j(t)$ for $j = 1, \dots, 9$.
\end{prop}

Consequently, if
\begin{equation}
  \label{finalb}
  \langle \omega_t, (\mathcal{N} \otimes I) \omega_t \rangle \le D e^{\alpha t}, \quad \langle \omega_t, (I \otimes \mathcal{N}) \omega_t \rangle \le D e^{\alpha t}, \quad \langle \omega_t, (\mathcal{N} \otimes \mathcal{N}) \omega_t \rangle \le D e^{\alpha t}
\end{equation}
for positive constants $D$ and $\alpha$, we conclude that
\[
  \Big| \text{\rm Tr} J \, \big( \gamma_{t}^{(1,1)} - | u_t \otimes v_t \rangle \langle u_t \otimes v_t | \big) \Big| \le \| J \| C e^{\gamma t} \left[ \frac{1}{\sqrt{N_1}} + \frac{1}{\sqrt{N_2}} \right]
\]
for positive constants $C$ and $\gamma$.
This estimate implies the desired inequality, as we explained earlier.
Therefore, to finish the proof of the theorem, we need to prove the estimates in \eqref{finalb}.
To to this, we use the following proposition, which we prove later in Section \ref{s:proofofp:part2}.

\begin{prop}
  \label{p:part2}
  For any $\psi \in \mathcal{F} \otimes \mathcal{F}$ and $j \in \N$, we have
  \begin{multline*}
    \langle U(t,s) \psi, (\mathcal{N} \otimes I + I \otimes \mathcal{N})^j U(t,s) \psi \rangle \\
    \le C_j \big\| [ (\mathcal{N} + 1) \otimes I + I \otimes (\mathcal{N} + 1) ]^{j+1} \psi \big\|^2 e^{\gamma_j |t-s|}.
  \end{multline*}
  Here $C_j$ and $\gamma_j$ are positive constants that do not depend on $N_1$ and $N_2$.
\end{prop}

Let us prove the estimates in \eqref{finalb}.
We start with two observations:
First, since $\omega_t = e^{-i \alpha(t)} U(t,0) \omega$, we have
\[
  \langle \omega_t, (\mathcal{N} \otimes I) \omega_t \rangle = \langle U(t,0) \omega, (\mathcal{N} \otimes I) U(t,0) \omega \rangle,
\]
and we have a similar equality for the other two quantities in \eqref{finalb}.
Secondly, we have $\mathcal{N} \Omega = 0$.
We now apply Proposition \ref{p:part2} with $\psi = \omega$, $s = 0$, $t > 0$, and $j = 1, 2$.
We obtain
\begin{gather*}
  \langle U(t,0) \omega, (\mathcal{N} \otimes I + I \otimes \mathcal{N}) U(t,0) \omega \rangle \le C_1 e^{\gamma_1 t}, \\
  \langle U(t,0) \omega, (\mathcal{N}^2 \otimes I + 2 \mathcal{N} \otimes \mathcal{N} + I \otimes \mathcal{N}^2) U(t,0) \omega \rangle \le C_2 e^{\gamma_2 t}.
\end{gather*}
On the left hand side of these inequalities, each term is a non-negative number.
Using these observations, we obtain the estimates in \eqref{finalb} for suitable constants $D$ and $\alpha$.
This completes the proof of Theorem~\ref{t:coherent}.
\qed

\section{Proof of Proposition \ref{p:part1}}
\label{s:proofofp:part1}

By \eqref{expgamma}, we have
\begin{equation}
  \label{sumofrj}
  \text{\rm Tr} J \big( \gamma_{t}^{(1,1)} - | u_t \otimes v_t \rangle \langle u_t \otimes v_t | \big) = \sum_{j=1}^{15} \int dx' dy' dx dy \, J(x',y';x,y) r_j(x,y;x',y').
\end{equation}
We want to estimate the absolute value of each term in this sum.
Recall that $\| u_t \| = 1$ and $\| v_t \| = 1$.
Using Cauchy-Schwarz inequality and the fact that $J$ is a bounded operator, we obtain
\[
  \begin{split}
    & \Big| \int dx' dy' dx dy \, J(x',y';x,y) r_1(x,y;x',y') \Big| \\
    % & = \Big| \int dx' dy' dx dy \, J(x',y';x,y) \overline{v_t}(y') u_t(x) v_t(y) \langle \omega_t, b_{x'}^* \omega_r \rangle \Big| \\
    & = N_1^{-1/2} \Big| \int dx' dy' \Big\langle v_t(y') b_{x'} \omega_t, \int dx dy \, J(x',y';x,y) u_t(x) v_t(y) \omega_t \Big\rangle \Big| \\
    & \le N_1^{-1/2} \int dx' dy' \| v_t(y') b_{x'} \omega_t \| \, \Big\| \int dx dy \, J(x',y';x,y) u_t(x) v_t(y) \omega_t \Big\| \\
    & \le N_1^{-1/2} \Big( \int dx' dy' |v_t(y')|^2 \| b_{x'} \omega_t \|^2 \Big)^{1/2} \\
    & \quad \times \Big( \int dx' dy' \Big\| \int dx dy \, J(x',y';x,y) u_t(x) v_t(y) \omega_t \Big\|^2 \Big)^{1/2} \\
    & \le N_1^{-1/2} \| J \| \| u_t \| \| v_t \|^2 \Big( \int dx' \| b_{x'} \omega_t \|^2 \Big)^{1/2} \\
    & = \| J \| N_1^{-1/2} \langle \omega_t, (\mathcal{N} \otimes I) \omega_t \rangle^{1/2}.
  \end{split}
\]
The same contribution to \eqref{sumofrj} arises from $r_3$ (the calculation is similar).
Hence the factor $2$ in $p_1$.

Similarly, the contribution arising from $r_2$ and $r_4$ is $\| J \| N_2^{-1/2} \langle \omega_t, (I \otimes \mathcal{N}) \omega_t \rangle^{1/2}$, which gives $p_2$.

The calculations to derive $p_3, \dots, p_9$ are similar and use the same ingredients: Cauchy-Schwarz inequality and the fact that $J$ is a bounded operator.
In the bound for \eqref{sumofrj}, the contribution $p_3$ arises from $r_5$ and $r_{10}$, the contribution $p_4$ arises from $r_7$ and $r_8$, the contribution $p_5$ arises from $r_6$, the contribution $p_6$ arises from $r_9$, the contribution $p_7$ arises from $r_{11}$ and $r_{13}$, the contribution $p_8$ arises from $r_{12}$ and $r_{14}$, and the contribution $p_9$ arises from $r_{15}$.
We omit the details.
This proves Proposition \ref{p:part1}.
\qed

\section{Proof of Proposition \ref{p:part2}}
\label{s:proofofp:part2}

To prove Proposition \ref{p:part2}, we will use another dynamics, denoted by $U^c(t,s)$, whose generator is similar to $\mathcal{L}(t)$ but contains a cutoff in the terms that are cubic in creation or annihilation operators.
Given positive constants $M_1$ and $M_2$, the cutoff forces the number of particles in the first factor of $\mathcal{F} \otimes \mathcal{F}$ to be smaller than $M_1$ and in the second factor to be smaller than $M_2$.
We refer to $U^c(t,s)$ as the truncated fluctuation dynamics, and we denote its generator by $\mathcal{L}^c(t)$.

\textbf{Truncated fluctuation dynamics.}
We define
\[
  \mathcal{L}^c(t) = \mathcal{H}_{N_1,N_2} + C_{N_1,N_2}^c(t) + Q_{N_1,N_2}(t) + R_{N_1,N_2}(t)
\]
with $\mathcal{H}_{N_1,N_2}$, $Q_{N_1,N_2}(t)$ and $R_{N_1,N_2}(t)$ as above and
\[
  \begin{split}
    C_{N_1,N_2}^c(t) & = \frac{1}{\sqrt{N_1}} \int dx dz \, V_1(x-z) b_x^* (u_t(z) \chi_1 b_z^* + \overline{u_t}(z) b_z \chi_1) b_x \\
    & \quad + \frac{1}{\sqrt{N_2}} \int dy dz \, V_2(y-z) c_y^* (v_t(z) \chi_2 c_z^* + \overline{v_t}(z) c_z \chi_2) c_y \\
    & \quad + \frac{\sqrt{N_1}}{N_1 + N_2} \int dx dy \, V_{12}(x-y) c_y^* ( u_t(x) \chi_1 b_x^* + \overline{u_t}(x) b_x \chi_1 ) c_y \\
    & \quad + \frac{\sqrt{N_2}}{N_1 + N_2} \int dx dy \, V_{12}(x-y) b_x^* ( v_t(y) \chi_2 c_y^* + \overline{v_t}(y) c_y \chi_2 ) b_x.
  \end{split}
\]
Here $\chi_1 = \chi(\mathcal{N} \le M_1) \otimes I$ and $\chi_2 = I \otimes \chi(\mathcal{N} \le M_2)$, where $\chi$ is a characteristic function.

For each $s \in \R$, let $U_{t,s}^c = U^c(t,s)$ be the time evolution defined by the equation
\[
  i \partial_t U_{t,s}^c = \mathcal{L}(t) U_{t,s}^c \qquad \text{with} \qquad U^c_{s,s} = I.
\]

The first ingredient to prove Proposition \ref{p:part2} is the following lemma:

\begin{lemma}
  \label{l:step1}
  For $\psi \in \mathcal{F} \otimes \mathcal{F}$ and $t,s \in \R$, we have
  \[
    \begin{split}
      \langle U^c_{t,s} \psi, ( (\mathcal{N} + 1) \otimes I & + I \otimes (\mathcal{N} + 1))^j U^c_{t,s} \psi \rangle \\
      & \le \langle \psi, ( (\mathcal{N} + 1) \otimes I + I \otimes (\mathcal{N} + 1) )^j \psi \rangle \\
      & \quad \times \exp \big( C_j (1 + \sqrt{M_1/N_1} + \sqrt{M_2/N_2}) |t-s| \big).
    \end{split}
  \]
  Here $C_j$ is a positive constant that does not depend on $N_1$ and $N_2$.
\end{lemma}

Before we prove this lemma, let us state and prove another lemma that will be used several times:

\begin{lemma}
  \label{l:Z}
  Suppose that $(N_1)$ and $(N_2)$ are sequences of positive integers such that $N_1 \to \infty$ and $N_2 \to \infty$ with
  \[
    \left| \frac{N_1}{N_1 + N_2} - c_1 \right| \le \frac{1}{N_2} \qquad \text{and} \qquad \left| \frac{N_2}{N_1 + N_2} - c_2 \right| \le \frac{1}{N_1}
  \]
  where $c_1 \geq 0$ and $c_2 \geq 0$ with $c_1 + c_2 = 1$.
  For $N_1$ and $N_2$ sufficiently large, there exists a constant $C > 0$ such that
  \[
    \frac{N_1+N_2}{N_1} < C, \qquad \frac{N_1+N_2}{N_2} < C, \qquad \frac{N_1}{N_2} < C, \qquad \frac{N_2}{N_1} < C, \qquad \frac{\sqrt{N_1} \sqrt{N_2}}{N_1+N_2} < C.
  \]
\end{lemma}

\begin{proof}
  For $k = 1, 2$, consider the sequence $a_k = N_k/(N_1+N_2)$.
  Since $a_k$ converges to $c_k \neq 0$, the sequence $1/a_k$ converges to $1/c_k$.
  Thus $1/a_k$ is bounded.
  This proves the first two estimates.
  The other bounds follow from the identities
  \begin{gather*}
    \frac{N_1}{N_2} = \bigg[ \frac{N_1}{N_1+N_2} - c_1 \bigg] \frac{N_1+N_2}{N_2} + c_1 \frac{N_1+N_2}{N_2}, \\
    \frac{\sqrt{N_1}\sqrt{N_2}}{N_1+N_2} = \sqrt{\frac{N_1}{N_2}} \bigg[ \frac{N_2}{N_1+N_2} - c_2 \bigg] + c_2 \sqrt{\frac{N_1}{N_2}}.
  \end{gather*}
\end{proof}

\begin{proof}[Proof of Lemma \ref{l:step1}]
  We fix $j$ and set $\mathcal{X} = ((\mathcal{N} + 1) \otimes I + I \otimes (\mathcal{N} + 1))^j$.
  We will estimate the time derivative of the expected value of $\mathcal{X}$ and then use Gronwall's Lemma.
  Define $\mathcal{X}_k = (\mathcal{N} + 1)^{j-k} \otimes (\mathcal{N} + 1)^k$.
  Then $\mathcal{X} = \sum_{k=0}^j \binom{j}{k} \mathcal{X}_k$.
  Calculating, we obtain
  \[
    \frac{d}{dt} \langle U^c_{t,s} \psi, \mathcal{X} U^c_{t,s} \psi \rangle = \sum_{k=0}^j \binom{j}{k} i \langle U^c_{t,s} \psi, [ \mathcal{L}^c(t), \mathcal{X}_k ] U^c_{t,s} \psi \rangle = \sum_{k=0}^j \binom{j}{k} \sum_{l=1}^{10} q_{l,k}
  \]
  where
  \begin{align*}
    q_{1,k} & = - \Im \int dx dz \, V_1(x-z) u_t(z) u_t(x) \langle U^c_{t,s} \psi, [b_x^* b_z^*, \mathcal{X}_k] U^c_{t,s} \psi \rangle, \\
    q_{2,k} & = - \Im \int dy dz \, V_2(y-z) v_t(z) v_t(y) \langle U^c_{t,s} \psi, [c_y^* c_z^*, \mathcal{X}_k] U^c_{t,s} \psi \rangle, \\
    q_{3,k} & = - \frac{\sqrt{N_1} \sqrt{N_2}}{N_1 + N_2} \, 2 \Im \!\! \int \! dx dy V_{12}(x-y) u_t(x) v_t(y) \langle U^c_{t,s} \psi, [ b_x^* c_y^*, \mathcal{X}_k ] U^c_{t,s} \psi \rangle, \\
    q_{4,k} & = - \frac{\sqrt{N_1} \sqrt{N_2}}{N_1 + N_2} \, 2 \Im \!\! \int \! dx dy V_{12}(x-y) u_t(x) \overline{v_t}(y) \langle U^c_{t,s} \psi, [ b_x^* c_y, \mathcal{X}_k ] U^c_{t,s} \psi \rangle, \\
    q_{5,k} & = - \frac{1}{\sqrt{N_1}} \, 2 \Im \int dx dz \, V_1(x-z) \overline{u_t}(z) \langle U^c_{t,s} \psi, [ b_x^* b_z \chi_1 b_x, \mathcal{X}_k ] U^c_{t,s} \psi \rangle, \\
    q_{6,k} & = - \frac{1}{\sqrt{N_2}} \, 2 \Im \int dy dz \, V_2(y-z) \overline{v_t}(z) \langle U^c_{t,s} \psi, [ c_y^* c_z \chi_2 c_y, \mathcal{X}_k ] U^c_{t,s} \psi \rangle, \\
    q_{7,k} & = - \frac{1}{\sqrt{N_1}} \, 2 \Im\int dx dy \, V_{12}(x-y) \overline{u_t}(x) \langle U^c_{t,s} \psi, [ c_y^* b_x \chi_1 c_y, \mathcal{X}_k ] U^c_{t,s} \psi \rangle, \\
    q_{8,k} & = - \frac{1}{\sqrt{N_2}} \, 2 \Im\int dx dy \, V_{12}(x-y) \overline{v_t}(y) \langle U^c_{t,s} \psi, [ b_x^* c_y \chi_2 b_x, \mathcal{X}_k ] U^c_{t,s} \psi \rangle, \\
    q_{9,k} & = - \sqrt{N_1} \left( \frac{N_2}{N_1 + N_2} - c_2 \right) \\
    & \qquad \times 2 \Im \int dx \, (V_{12} * |v_t|^2)(x) u_t(x) \langle U^c_{t,s} \psi, [ b_x^*, \mathcal{X}_k ] U^c_{t,s} \psi \rangle, \\
    q_{10,k} & = - \sqrt{N_2} \left( \frac{N_1}{N_1 + N_2} - c_1 \right) \\
    & \qquad \times 2 \Im \int dy \, (V_{12} * |u_t|^2)(y) v_t(y) \langle U^c_{t,s} \psi, [ c_y^*, \mathcal{X}_k ] U^c_{t,s} \psi \rangle.
  \end{align*}
  We observe that $q_{l,k} = q_{l,k}(t)$ for $l = 1, \dots, 10$.
  For each $l$, we will prove that
  \begin{equation}
    \label{leftto1}
    \sum_{k=0}^j \binom{j}{k} q_{l,k} \le C'_j (1 + \sqrt{M_1/N_1} + \sqrt{M_2/N_2}) \langle U^c_{t,s} \psi, \mathcal{X} U^c_{t,s} \psi \rangle.
  \end{equation}
  We then obtain
  \[
    \frac{d}{dt} \langle U^c_{t,s} \psi, \mathcal{X} U^c_{t,s} \psi \rangle \le C_j (1 + \sqrt{M_1/N_1} + \sqrt{M_2/N_2}) \langle U^c_{t,s} \psi, \mathcal{X} U^c_{t,s} \psi \rangle.
  \]
  By Gronwall's Lemma, this implies the desired estimate.

  We are left to prove the inequality \eqref{leftto1}.
  We need to calculate the commutators in the expressions for $q_{l,k}$.
  First, using $a_x \mathcal{N} = (\mathcal{N} + 1) a_x$ and $a_x^* \mathcal{N} = (\mathcal{N} - 1) a_x^*$, we obtain
  \begin{align*}
    [a_x^*, (\mathcal{N} + 1)^p] & = \sum_{q=0}^{p-1} \binom{p}{q} (-1)^q (\mathcal{N} + 1)^q a_x^*, \\
    [a_x, (\mathcal{N} + 1)^p] & = \sum_{q=0}^{p-1} \binom{p}{q} (\mathcal{N} + 1)^q a_x.
  \end{align*}
  Consequently, we get
  \[
    \begin{split}
      & [a_x^* a_y^*, (\mathcal{N} + 1)^p] \\
      & = \sum_{q=0}^{p-1} \binom{p}{q} (-1)^q \big( a_x^* (\mathcal{N} + 1)^q a_y^* + (\mathcal{N} + 1)^q a_x^* a_y^* \big) \\
      & = \sum_{q=0}^{p-1} \binom{p}{q} (-1)^q \big( \mathcal{N}^{q/2} a_x^* a_y^* (\mathcal{N} + 2)^{q/2} + (\mathcal{N} + 1)^{q/2} a_x^* a_y^* (\mathcal{N} + 3)^{q/2} \big)
    \end{split}
  \]
  and
  \[
    [a_x, (\mathcal{N} + 1)^p] = \sum_{q=0}^{p-1} \binom{p}{q} (\mathcal{N} + 1)^{q/2} a_x \mathcal{N}^{q/2}.
  \]
  Using the above formulae, we find that
  \begin{equation}
    \label{commB}
    [b_x^*, (\mathcal{N} + 1)^p \otimes (\mathcal{N} + 1)^r] = \sum_{q=0}^{p-1} \binom{p}{q} (-1)^q (\mathcal{N} + 1)^q \otimes (\mathcal{N} + 1)^r b_x^*,
  \end{equation}
  \begin{equation}
    [b_x, (\mathcal{N} + 1)^p \otimes (\mathcal{N} + 1)^r] = \sum_{q=0}^{p-1} \binom{p}{q} (\mathcal{N} + 1)^q \otimes (\mathcal{N} + 1)^r b_x,
  \end{equation}
  \begin{equation}
    \begin{split}
      & [b_x^* b_z^*, (\mathcal{N} + 1)^p \otimes (\mathcal{N} + 1)^r] \\
      & = \sum_{l=0}^{p-1} \binom{p}{l} (-1)^l \mathcal{N}^{l/2} \otimes (\mathcal{N} + 1)^{r/2} b_x^* b_z^* (\mathcal{N} + 2)^{l/2} \otimes (\mathcal{N} + 1)^{r/2} \\
      & \quad + \sum_{l=0}^{p-1} \binom{p}{l} (-1)^l (\mathcal{N} + 1)^{l/2} \otimes (\mathcal{N} + 1)^{r/2} b_x^* b_z^* (\mathcal{N} + 3)^{l/2} \otimes (\mathcal{N} + 1)^{r/2},
    \end{split}
  \end{equation}
  \begin{equation}
    \label{commBC}
    \begin{split}
      & [b_x^* c_y^*, (\mathcal{N} + 1)^p \otimes (\mathcal{N} + 1)^r ] \\
      & = \sum_{l=0}^{p-1} \binom{p}{l} (-1)^l (\mathcal{N} + 1)^{l/2} \otimes (\mathcal{N} + 1)^{r/2} b_x^* c_y^* (\mathcal{N} + 2)^{l/2} \otimes (\mathcal{N} + 2)^{r/2} \\
      & \quad + \sum_{l=0}^{r-1} \binom{r}{l} (-1)^l \mathcal{N}^{p/2} \otimes (\mathcal{N} + 1)^{l/2} b_x^* c_y^* (\mathcal{N} + 1)^{p/2} \otimes (\mathcal{N} + 2)^{l/2}.
    \end{split}
  \end{equation}
  We have similar expressions for the commutators
  \begin{equation}
    \label{commC}
    \begin{alignedat}{2}
      & [c_y^*, (\mathcal{N} + 1)^p \otimes (\mathcal{N} + 1)^r], & & \qquad [c_y, (\mathcal{N} + 1)^p \otimes (\mathcal{N} + 1)^r], \\
      & [c_y^* c_z^*, (\mathcal{N} + 1)^p \otimes (\mathcal{N} + 1)^r], & & \qquad [b_x^* c_y^*, (\mathcal{N} + 1)^p \otimes (\mathcal{N} + 1)^r ].
    \end{alignedat}
  \end{equation}

  We can now prove \eqref{leftto1}.
  We want to estimate
  \[
    \sum_{k=0}^j \binom{j}{k} q_{l,k}
  \]
  for $l = 1, \dots, 10$.
  By symmetry with respect to exchanging operators $b^{(\cdot)}_{(\cdot)}$ and $c^{(\cdot)}_{(\cdot)}$, there are only six types of terms to consider, namely, the terms corresponding to $l$ belonging to the sets $\setl{1,2}$, $\setl{3}$, $\setl{4}$, $\setl{5,6}$, $\setl{7,8}$ and $\setl{9,10}$, respectively.

  Let us consider $l = 3$.
  Write
  \[
    g_3(t) = \Im \int dx dy \, V_{12}(x-y) u_t(x) v_t(y) \langle U^c_{t,s} \psi, [ b_x^* c_y^*, \mathcal{X}_k ] U^c_{t,s} \psi \rangle.
  \]
  Using Lemma \ref{l:Z}, for $N_1$ and $N_2$ sufficiently large, we obtain
  \[
    \left| \sum_{k=0}^j \binom{j}{k} q_{3,k} \right| = \left| -2 \frac{\sqrt{N_1} \sqrt{N_2}}{N_1 + N_2} \sum_{k=0}^j \binom{j}{k} g_3(t) \right| \le C \sum_{k=0}^j \binom{j}{k} |g_3(t)|.
  \]
  Using \eqref{commBC}, we get
  \[
    \begin{split}
      g_3(t) & = \sum_{l=0}^{j-k-1} \binom{j-k}{l} (-1)^{l+1} \Im \int dx \, u_t(x) \big \langle b_x (\mathcal{N}+1)^{l/2} \otimes (\mathcal{N}+1)^{k/2} U^c_{t,s} \psi, \\
      & \qquad \qquad \qquad \qquad c^*(V_{12}(x-\cdot) v_t(\cdot)) \big( (\mathcal{N}+1)^{l/2} \otimes (\mathcal{N}+1)^{k/2} \big) U^c_{t,s} \psi \big \rangle \\
      & \quad + \sum_{l=0}^{k-1} \binom{k}{l} (-1)^{l+1} (-1)^{l+1} \Im \int dx \, u_t(x) \big \langle b_x \mathcal{N}^{(j-k)/2} \otimes (\mathcal{N}+1)^{l/2} U^c_{t,s} \psi, \\
      & \qquad \qquad \qquad \qquad c^*(V_{12}(x-\cdot) v_t(\cdot)) \big( (\mathcal{N}+1)^{(j-k)/2} \otimes (\mathcal{N}+2)^{l/2} \big) U^c_{t,s} \psi \big \rangle.
    \end{split}
  \]
  Using the Cauchy-Schwarz inequality and Lemma \ref{l:relbN}, we find
  \[
    \begin{split}
      |g_3(t)| & \le \sum_{l=0}^{j-k-1} \binom{j-k}{l} \int dx \, |u_t(x)| \| b_x (\mathcal{N}+1)^{l/2} \otimes (\mathcal{N}+1)^{k/2} U^c_{t,s} \psi \| \\
      & \qquad \qquad \times \| c^*(V_{12}(x-\cdot) v_t(\cdot)) (\mathcal{N}+1)^{l/2} \otimes (\mathcal{N}+1)^{k/2} \, U^c_{t,s} \psi \| \\
      & \quad + \sum_{l=0}^{k-1} \binom{k}{l} \int dx \, |u_t(x)| \| b_x \mathcal{N}^{(j-k)/2} \otimes (\mathcal{N}+1)^{l/2} U^c_{t,s} \psi \| \\
      & \qquad \qquad \times \| c^*(V_{12}(x-\cdot) v_t(\cdot)) (\mathcal{N}+1)^{(j-k)/2} \otimes (\mathcal{N}+2)^{l/2} U^c_{t,s} \psi \| \\
      & \le \sum_{l=0}^{j-k-1} \binom{j-k}{l} \left( \int dx \, \| b_x (\mathcal{N}+1)^{l/2} \otimes (\mathcal{N}+1)^{k/2} U^c_{t,s} \psi \|^2 \right)^{1/2} \\
      & \qquad \qquad \times C^{1/2} \| (\mathcal{N}+1)^{l/2} \otimes (\mathcal{N}+1)^{(k+1)/2} \, U^c_{t,s} \psi \| \\
      & \quad + \sum_{l=0}^{k-1} \binom{k}{l} \left( \int dx \, \| b_x \mathcal{N}^{(j-k)/2} \otimes (\mathcal{N}+1)^{l/2} U^c_{t,s} \psi \|^2 \right)^{1/2} \\
      & \qquad \qquad \times C^{1/2} \| (\mathcal{N}+1)^{(j-k)/2} \otimes (\mathcal{N}+2)^{(l+1)/2} U^c_{t,s} \psi \|.
    \end{split}
  \]
  Here, we have used that
  \[
    \int dx \, \| V_{12}(x - \cdot) v_t(\cdot) \|^2 \, |u_t(x)|^2 \le K \int dy \, |v_t(y)|^2 \int dz \, (|u_t(z)|^2 + |\nabla u_t(z)|^2 ) \le C.
  \]
  This estimate follows from the hypothesis on $V_{12}$ (using integration by parts).
  Now, using the Cauchy-Schwarz inequality again and rewriting, we obtain
  \[
    \begin{split}
      |g_3(t)| & \le C_{j-k} (1-\delta_{kj}) \big( \langle U^c_{t,s} \psi, (\mathcal{N}+1)^{j-k} \otimes (\mathcal{N}+1)^k U^c_{t,s} \psi \rangle \\
      & \qquad \qquad \qquad + \langle U^c_{t,s} \psi, (\mathcal{N}+1)^{j-k-1} \otimes (\mathcal{N}+1)^{k+1} \, U^c_{t,s} \psi \rangle \big) \\
      & \quad + C_k (1-\delta_{k0}) \big( \langle U^c_{t,s} \psi, (\mathcal{N}+1)^{j-k+1} \otimes (\mathcal{N}+1)^{k-1} U^c_{t,s} \psi \rangle \\
      & \qquad \qquad \qquad + \langle U^c_{t,s} \psi, (\mathcal{N}+1)^{j-k} \otimes (\mathcal{N}+1)^k U^c_{t,s} \psi \rangle \big)
    \end{split}
  \]
  for some constants $C_{j-k}$ and $C_k$, where $\delta_{jk}$ denotes the Kronecker delta.
  Finally, we arrive at
  \[
    \begin{split}
      \sum_{k=0}^j \binom{j}{k} |g_3(t)| & \le C_j' \sum_{k=0}^j \binom{j}{k} \langle U^c_{t,s} \psi, (\mathcal{N}+1)^{j-k} \otimes (\mathcal{N}+1)^k U^c_{t,s} \psi \rangle \\
      & = C_j' \langle U^c_{t,s} \psi, \big( (\mathcal{N}+1) \otimes I + I \otimes (\mathcal{N}+1) \big)^j U^c_{t,s} \psi \rangle \\
      & = C_j' \langle U^c_{t,s} \psi, \mathcal{X} U^c_{t,s} \psi \rangle.
    \end{split}
  \]
  This upper bound contributes to the first term on the right hand side of \eqref{leftto1}.

  Let us consider $l \in \{1, 2\}$ and $l=4$.
  The calculations in these cases are very similar to the calculations in the case $l=3$.
  Write
  \begin{align*}
    g_1(t) & = \Im \int dx dz \, V_{1}(x-z) u_t(z) u_t(x) \langle U^c_{t,s} \psi, [ b_x^* b_z^*, \mathcal{X}_k ] U^c_{t,s} \psi \rangle, \\
    g_2(t) & = \Im \int dy dz \, V_{2}(y-z) v_t(z) v_t(y) \langle U^c_{t,s} \psi, [ c_y^* c_z^*, \mathcal{X}_k ] U^c_{t,s} \psi \rangle, \\
    g_4(t) & = \Im \int dx dy \, V_{12}(x-y) u_t(x) \overline{v_t}(y) \langle U^c_{t,s} \psi, [ b_x^* c_y, \mathcal{X}_k ] U^c_{t,s} \psi \rangle.
  \end{align*}
  Using Lemma \ref{l:Z}, for $N_1$ and $N_2$ sufficiently large, we obtain
  \begin{align*}
    & \left| \sum_{k=0}^j \binom{j}{k} q_{1,k} \right| = \left| \sum_{k=0}^j \binom{j}{k} g_1(t) \right| \le C \sum_{k=0}^j \binom{j}{k} |g_1(t)|, \\
    & \left| \sum_{k=0}^j \binom{j}{k} q_{2,k} \right| = \left| \sum_{k=0}^j \binom{j}{k} g_2(t) \right| \le C \sum_{k=0}^j \binom{j}{k} |g_2(t)|, \\
    & \left| \sum_{k=0}^j \binom{j}{k} q_{4,k} \right| = \left| -2 \frac{\sqrt{N_1} \sqrt{N_2}}{N_1 + N_2} \sum_{k=0}^j \binom{j}{k} g_4(t) \right| \le C \sum_{k=0}^j \binom{j}{k} |g_4(t)|.
  \end{align*}
  By proceeding similarly as in the case $l=3$, we get
  \[
    \sum_{k=0}^j \binom{j}{k} |g_l(t)| \le C_j' \langle U^c_{t,s} \psi, \mathcal{X} U^c_{t,s} \psi \rangle
  \]
  for $l = 1, 2, 4$.

  Let us consider $l \in \{5, 6, 7, 8 \}$.
  Write
  \begin{align*}
    g_5(t) & = \Im \int dx dz \, V_{1}(x-z) \overline{u_t}(z) \langle U^c_{t,s} \psi, [ b_x^* b_z \chi_1 b_x, \mathcal{X}_k ] U^c_{t,s} \psi \rangle, \\
    g_6(t) & = \Im \int dy dz \, V_{2}(y-z) \overline{v_t}(z) \langle U^c_{t,s} \psi, [ c_y^* c_z \chi_2 c_y, \mathcal{X}_k ] U^c_{t,s} \psi \rangle, \\
    g_7(t) & = \Im \int dx dy \, V_{12}(x-y) \overline{u_t}(x) \langle U^c_{t,s} \psi, [ c_y^* b_x \chi_1 c_y, \mathcal{X}_k ] U^c_{t,s} \psi \rangle, \\
    g_8(t) & = \Im \int dx dy \, V_{12}(x-y) \overline{v_t}(y) \langle U^c_{t,s} \psi, [ b_x^* c_y \chi_2 b_x, \mathcal{X}_k ] U^c_{t,s} \psi \rangle.
  \end{align*}
  For $N_1$ and $N_2$ sufficiently large, we obtain
  \begin{align*}
    & \left| \sum_{k=0}^j \binom{j}{k} q_{5,k} \right| = \left| \frac{2}{\sqrt{N_1}} \sum_{k=0}^j \binom{j}{k} g_5(t) \right| \le \frac{C}{\sqrt{N_1}} \sum_{k=0}^j \binom{j}{k} |g_5(t)|, \\
    & \left| \sum_{k=0}^j \binom{j}{k} q_{6,k} \right| = \left| \frac{2}{\sqrt{N_2}} \sum_{k=0}^j \binom{j}{k} g_6(t) \right| \le \frac{C}{\sqrt{N_2}} \sum_{k=0}^j \binom{j}{k} |g_6(t)|, \\
    & \left| \sum_{k=0}^j \binom{j}{k} q_{7,k} \right| = \left| \frac{2}{\sqrt{N_1}} \sum_{k=0}^j \binom{j}{k} g_7(t) \right| \le \frac{C}{\sqrt{N_1}} \sum_{k=0}^j \binom{j}{k} |g_7(t)|, \\
    & \left| \sum_{k=0}^j \binom{j}{k} q_{8,k} \right| = \left| \frac{2}{\sqrt{N_2}} \sum_{k=0}^j \binom{j}{k} g_8(t) \right| \le \frac{C}{\sqrt{N_2}} \sum_{k=0}^j \binom{j}{k} |g_8(t)|.
  \end{align*}
  We will estimate the contribution arising from $|g_8(t)|$.
  The estimates related to $|g_5(t)|$, $|g_6(t)|$ and $|g_7(t)|$ are very similar.
  First, since $[\chi_2, b_x] = 0$ and
  \[
    [\chi_2, (\mathcal{N}+1)^{j-k} \otimes (\mathcal{N}+1)^k] = 0,
  \]
  we have
  \[
    \begin{split}
      & [b_x^* c_y \chi_2 b_x, \mathcal{X}] \\
      & = [b_x^*, (\mathcal{N}+1)^{j-k} \otimes (\mathcal{N}+1)^k] c_y b_x \chi_2 + b_x^* [c_y, (\mathcal{N}+1)^{j-k} \otimes (\mathcal{N}+1)^k] b_x \chi_2 \\
      & \qquad + b_x^* c_y [b_x, (\mathcal{N}+1)^{j-k} \otimes (\mathcal{N}+1)^k] \chi_2.
    \end{split}
  \]
  Using \eqref{commB}-\eqref{commC}, we obtain
  \[
    \begin{split}
      & [b_x^* c_y \chi_2 b_x, \mathcal{X}] \\
      & = \sum_{l=0}^{j-k-1} \binom{j-k}{l} (-1)^l (\mathcal{N}+1)^{l/2} \otimes (\mathcal{N}+1)^{k/2} b_x^* b_x c_y \chi_2 (\mathcal{N}+1)^{l/2} \otimes \mathcal{N}^{k/2} \\
      & \quad + \sum_{l=0}^{k-1} \binom{k}{l} \mathcal{N}^{(j-k)/2} \otimes (\mathcal{N}+1)^{l/2} b_x^* b_x c_y \chi_2 \mathcal{N}^{(j-k)/2} \otimes \mathcal{N}^{l/2} \\
      & \quad + \sum_{l=0}^{j-k-1} \binom{j-k}{l} \mathcal{N}^{l/2} \otimes (\mathcal{N}+2)^{k/2} b_x^* b_x c_y \chi_2 \mathcal{N}^{l/2} \otimes (\mathcal{N}+1)^{k/2}.
    \end{split}
  \]
  Now, we substitute this formula into the expression for $g_8(t)$ and apply Cauchy-Schwarz inequality and Lemma \ref{l:relbN}.
  We get
  \[
    \begin{split}
      & |g_8(t)| \\
      & \le \sum_{l=0}^{j-k-1} \binom{j-k}{l} \int dx \, \| b_x (\mathcal{N}+1)^{l/2} \otimes (\mathcal{N}+1)^{k/2} U^c_{t,s} \psi \| \, \| c(V_{12}(x-\cdot) v_t(\cdot)) \chi_2 \| \\
      & \qquad \qquad \times \| b_x (\mathcal{N}+1)^{l/2} \otimes \mathcal{N}^{k/2} U^c_{t,s} \psi \| \\
      & \quad + \sum_{l=0}^{k-1} \binom{k}{l} \int dx \, \| b_x \mathcal{N}^{(j-k)/2} \otimes (\mathcal{N}+1)^{l/2} U^c_{t,s} \psi \| \, \| c(V_{12}(x-\cdot) v_t(\cdot)) \chi_2 \| \\
      & \qquad \qquad \times \| b_x \mathcal{N}^{(j-k)/2} \otimes \mathcal{N}^{l/2} U^c_{t,s} \psi \| \\
      & \quad + \sum_{l=0}^{j-k-1} \binom{j-k}{l} \int dx \, \| b_x \mathcal{N}^{l/2} \otimes (\mathcal{N}+2)^{k/2} U^c_{t,s} \psi \| \, \| c(V_{12}(x-\cdot) v_t(\cdot)) \chi_2 \| \\
      & \qquad \qquad \times \| b_x \mathcal{N}^{l/2} \otimes (\mathcal{N}+1)^{k/2} U^c_{t,s} \psi \|.
    \end{split}
  \]
  Notice that
  \[
    \| c(V_{12}(x-\cdot) v_t(\cdot)) \chi_2 \| \le \sqrt{M_2} \| V_{12}(x-\cdot) v_t(\cdot) \|
  \]
  because
  \[
    \begin{split}
      \| c(V_{12}(x-\cdot) v_t(\cdot)) \chi_2 \psi \| & \le \| V_{12}(x-\cdot) v_t(\cdot) \| \, \| (I \otimes \mathcal{N}^{1/2}) (I \otimes \chi(\mathcal{N} \le M_2)) \psi \| \\
      & \le \| V_{12}(x-\cdot) v_t(\cdot) \| \sqrt{M_2} \| \psi \|.
    \end{split}
  \]
  Furthermore
  \[
    \sup_{x \in \R^3} \| V_{12}(x-\cdot) v_t(\cdot) \| \le K \int dy \, (|v_t(y)|^2 + |\nabla v_t(y)|^2) \le C.
  \]
  Recall that
  \[
    \int b_x^* b_x = \mathcal{N} \otimes I.
  \]
  Using these observations and applying the Cauchy-Schwarz inequality, we find
  \[
    \begin{split}
      |g_8(t)| & \le C \sqrt{M_2} \sum_{l=0}^{j-k-1} \binom{j-k}{l} \langle U^c_{t,s} \psi, (\mathcal{N}+1)^{l+1} \otimes (\mathcal{N}+1)^k U^c_{t,s} \psi \rangle \\
      & \quad + C \sqrt{M_2} \sum_{l=0}^{k-1} \binom{k}{l} \langle U^c_{t,s} \psi, (\mathcal{N}+1)^{j-k+1} \otimes (\mathcal{N}+1)^l U^c_{t,s} \psi \rangle \\
      & \quad + C \sqrt{M_2} \sum_{l=0}^{j-k-1} \binom{j-k}{l} \langle U^c_{t,s} \psi, (\mathcal{N}+1)^{l+1} \otimes (\mathcal{N}+1)^k U^c_{t,s} \psi \rangle \\
      & \le C_{j-k} \sqrt{M_2} \langle U^c_{t,s} \psi, (\mathcal{N}+1)^{j-k} \otimes (\mathcal{N}+1)^k U^c_{t,s} \psi \rangle \\
      & \quad + C_{j-k} (1-\delta_{k0}) \sqrt{M_2} \langle U^c_{t,s} \psi, (\mathcal{N}+1)^{j-k+1} \otimes (\mathcal{N}+1)^{k-1} U^c_{t,s} \psi \rangle. \\
    \end{split}
  \]
  Finally, by estimating some terms from above, we get
  \[
    \begin{split}
      & \frac{1}{\sqrt{N_2}} \sum_{k=0}^j \binom{j}{k} |g_8(t)| \\
      & \le C_j' \sqrt{M_2} \sum_{k=0}^j \binom{j}{k} \langle U^c_{t,s} \psi, (\mathcal{N}+1)^{j-k} \otimes (\mathcal{N}+1)^k U^c_{t,s} \psi \rangle \\
      & = C_j' \sqrt{M_2/N_2} \langle U^c_{t,s} \psi, \big( (\mathcal{N}+1) \otimes I + I \otimes (\mathcal{N}+1) \big)^j U^c_{t,s} \psi \rangle \\
      & = C_j' \sqrt{M_2/N_2} \langle U^c_{t,s} \psi, \mathcal{X} U^c_{t,s} \psi \rangle.
    \end{split}
  \]
  By similar calculations, we obtain similar estimates arising from $|g_5(t)|$, $|g_6(t)|$ and $|g_7(t)|$ (with $M_2$ and $N_2$ replaced by $M_1$ and $N_2$, accordingly).

  Finally, let us consider $l \in \{9, 10\}$.
  Write
  \begin{align*}
    g_9(t) & = \Im \int dx \, (V_{12} * |v_t|^2)(x) u_t(x) \langle U^c_{t,s} \psi, [ b_x^*, \mathcal{X}_k ] U^c_{t,s} \psi \rangle, \\
    g_{10}(t) & = \Im \int dy \, (V_{12} * |u_t|^2)(y) v_t(y) \langle U^c_{t,s} \psi, [ c_y^*, \mathcal{X}_k ] U^c_{t,s} \psi \rangle.
  \end{align*}
  For $N_1$ and $N_2$ sufficiently large, we have
  \begin{align*}
    \left| \sum_{k=0}^j \binom{j}{k} q_{9,k} \right| & = \left| -2 \sqrt{N_1} \left( \frac{N_2}{N_1 + N_2} - c_2 \right) \sum_{k=0}^j \binom{j}{k} g_9(t) \right| \le C \sum_{k=0}^j \binom{j}{k} |g_9(t)|, \\
    \left| \sum_{k=0}^j \binom{j}{k} q_{10,k} \right| & = \left| -2 \sqrt{N_2} \left( \frac{N_1}{N_1 + N_2} - c_1 \right) \sum_{k=0}^j \binom{j}{k} g_{10}(t) \right| \le C \sum_{k=0}^j \binom{j}{k} |g_{10}(t)|.
  \end{align*}
  Using \eqref{commB}, we get
  \[
    \begin{split}
      g_9(t) & = \sum_{l=0}^{j-k-1} \binom{j-k}{l} (-1)^l \Im \int dx \, (V_{12}*|v_t|^2)(x) u_t(x) \\
      & \qquad \times \langle b_x (\mathcal{N}+1)^{l/2} \otimes (\mathcal{N}+1)^{k/2} U^c_{t,s} \psi, (\mathcal{N}+1)^{l/2} \otimes (\mathcal{N}+1)^{k/2} U^c_{t,s} \psi \rangle.
    \end{split}
  \]
  Now, applying the Cauchy-Schwarz inequality twice, we obtain
  \[
    \begin{split}
      |g_9(t)| & \le C^{1/2} \sum_{l=0}^{j-k-1} \binom{j-k}{l} \left( \int dx \, \| b_x (\mathcal{N}+1)^{l/2} \otimes (\mathcal{N}+1)^{k/2} U^c_{t,s} \psi \|^2 \right)^{1/2} \\
      & \qquad \qquad \times \| (\mathcal{N}+1)^{l/2} \otimes (\mathcal{N}+1)^{k/2} U^c_{t,s} \psi \| \\
      & \le C^{1/2} \sum_{l=0}^{j-k-1} \big( \| (\mathcal{N}+1)^{(l+1)/2} \otimes (\mathcal{N}+1)^{k/2} U^c_{t,s} \psi \| \\
      & \qquad \qquad \times \| (\mathcal{N}+2)^{l/2} \otimes (\mathcal{N}+1)^{k/2} U^c_{t,s} \psi \| \\
      & \le C'_j \langle U^c_{t,s} \psi, (\mathcal{N}+1)^{j-k} \otimes (\mathcal{N}+1)^k U^c_{t,s} \psi \rangle,
    \end{split}
  \]
  where we used that
  \[
    \int dx \, |(V_{12}*|v_t|^2)(x)|^2 |u_t(x)|^2 \le C.
  \]
  Therefore, by proceeding similarly as above, we get
  \[
    \begin{split}
      \sum_{k=0}^j \binom{j}{k} |g_9(t)| & \le C_j' \langle U^c_{t,s} \psi, \big( (\mathcal{N}+1) \otimes I + I \otimes (\mathcal{N}+1) \big)^j U^c_{t,s} \psi \rangle \\
      & = C_j' \langle U^c_{t,s} \psi, \mathcal{X} U^c_{t,s} \psi \rangle.
    \end{split}
  \]
  Furthermore, we have a similar estimate arising from $|g_{10}(t)|$.
  This completes the proof of \eqref{leftto1}, which proves Lemma \ref{l:step1}.
\end{proof}

\textbf{A priori estimates.}
As a first step to prove Proposition \ref{p:part2}, we introduced the truncated fluctuation dynamics $U^c_{t,s}$ and proved Lemma \ref{l:step1}.
Now, we want to compare the evolution $U_{t,s}$ with the evolution $U^c_{t,s}$.
To do this, we first need some a-priori estimates on the growth of the number of particles with respect to the fluctuation dynamics.
This is the contents of the next lemma.

\begin{lemma}
  \label{l:step2}
  For $\psi \in \mathcal{F} \otimes \mathcal{F}$ and $t,s \in \R$, we have
  \begin{equation}
    \label{step2a}
    \langle \psi, U(t,s)^* ( \mathcal{N} \otimes I + I \otimes \mathcal{N} ) U(t,s) \psi \rangle \le 24 \langle \psi, (\mathcal{N} \otimes I + I \otimes \mathcal{N} + N_1 + N_2 + 1) \psi \rangle.
  \end{equation}
  Furthermore, for $j \in \N$,
  \begin{align}
    & \langle \psi, U(t,s)^* ( \mathcal{N} \otimes I + I \otimes \mathcal{N} )^{2j} U(t,s) \psi \rangle \label{step2b}\\
    & \qquad \qquad \le C_j \langle \psi, (\mathcal{N} \otimes I + I \otimes \mathcal{N} + N_1 + N_2)^{2j} \psi \rangle, \notag \\
    & \langle \psi, U(t,s)^* ( \mathcal{N} \otimes I + I \otimes \mathcal{N} )^{2j+1} U(t,s) \psi \rangle \label{step2c} \\
    & \qquad \qquad \le D_j \langle \psi, (\mathcal{N} \otimes I + I \otimes \mathcal{N} + N_1 + N_2)^{2j+1} (\mathcal{N} \otimes I + I \otimes \mathcal{N} + 1) \psi \rangle \notag
  \end{align}
  for appropriate constants $C_j$ and $D_j$.
\end{lemma}

\begin{proof}
  We will use the shorthand $\mathcal{H} = \mathcal{H}_{N_1,N_2}$.
  Using Lemma \ref{l:weyl}(b), we calculate
  \[
    \begin{split}
      & U^*(t,s) (\mathcal{N} \otimes I) U(t,s) \\
      & = \int dx \, U^*(t,s) b_x^* b_x U(t,s) \\
      & = \int dx \, \mathcal{W}_s^* e^{i(t-s)\mathcal{H}} (b_x^* - \sqrt{N_1} \overline{u_t}(x)) (b_x - \sqrt{N_1} u_t(x)) e^{-i(t-s)\mathcal{H}} \mathcal{W}_s \\
      & = \mathcal{W}_s^* ( \mathcal{N} \otimes I - \sqrt{N_1} e^{i(t-s) \mathcal{H}} (\phi(u_t) \otimes I) e^{-i(t-s) \mathcal{H}} + N_1 ) \mathcal{W}_s.
    \end{split}
  \]
  Similarly, we obtain
  \[
    U^*(t,s) (I \otimes \mathcal{N}) U(t,s) = \mathcal{W}_s^* ( I \otimes \mathcal{N} - \sqrt{N_2} e^{i(t-s) \mathcal{H}} (I \otimes \phi(v_t)) e^{-i(t-s) \mathcal{H}} + N_2 ) \mathcal{W}_s.
  \]

  Now, using the Cauchy-Schwarz inequality, Lemma \ref{l:relbN}, and Lemma \ref{l:weyl}, we have
  \[
    \begin{split}
      & \langle \psi, U^*(t,s) (\mathcal{N} \otimes I) U(t,s) \psi \rangle \\
      & = \langle \psi, (\mathcal{N} \otimes I + 2N_1) \psi \rangle - \sqrt{N_1} \langle \psi, (\phi(u_s) \otimes I) \psi \rangle \\
      & \quad - \sqrt{N_1} \langle e^{-i(t-s) \mathcal{H}} \mathcal{W}_s \psi, (\phi(u_t) \otimes I) e^{-i(t-s) \mathcal{H}} \mathcal{W}_s \psi \rangle \\
      & \le \langle \psi, (\mathcal{N} \otimes I + 2N_1) \psi \rangle + \sqrt{N_1} \| (\phi(u_s) \otimes I) \psi \| \, \| \psi \| \\
      & \quad + \sqrt{N_1} \| \psi \| \, \| (\phi(u_t) \otimes I) e^{-i(t-s) \mathcal{H}} \mathcal{W}_s \psi \| \\
      & \le \langle \psi, (\mathcal{N} \otimes I + 2N_1) \psi \rangle + 2N_1 \langle \psi, \psi \rangle + 2 \langle \psi, ((\mathcal{N}+1) \otimes I) \psi \rangle \\
      & \quad + 2 N_1 \langle \psi, \psi \rangle + 2 \langle \psi, \mathcal{W}_s^* e^{i(t-s) \mathcal{H}} ((\mathcal{N}+1) \otimes I) e^{-i(t-s) \mathcal{H}} \mathcal{W}_s \psi \rangle \\
      & \le 12 \langle \psi, (\mathcal{N} \otimes I + N_1 + 1) \psi \rangle.
    \end{split}
  \]
  Similarly, we obtain
  \[
    \langle \psi, U^*(t,s) (I \otimes \mathcal{N}) U(t,s) \psi \rangle \le 12 \langle \psi, (I \otimes \mathcal{N} + N_2 + 1) \psi \rangle.
  \]
  Therefore
  \[
    \langle \psi, U^*(t,s) (\mathcal{N} \otimes I + I \otimes \mathcal{N}) U(t,s) \psi \rangle \le 24 \langle \psi, (\mathcal{N} \otimes I + I \otimes \mathcal{N} + N_1 + N_2 + 1) \psi \rangle,
  \]
  which proves \eqref{step2a}.

  To prove \eqref{step2b}, we will use induction.
  Set
  \[
    Z_{t,s} = \mathcal{N} \otimes I + I \otimes \mathcal{N} - e^{i(t-s) \mathcal{H}} ( \sqrt{N_1} \phi(u_t) \otimes I + \sqrt{N_2} I \otimes \phi(v_t) ) e^{-i(t-s) \mathcal{H}} + N_1 + N_2.
  \]
  Then we can write
  \[
    U^*(t,s) (\mathcal{N} \otimes I + I \otimes \mathcal{N}) U(t,s) = \mathcal{W}_s^* Z_{t,s} \mathcal{W}_s.
  \]

  We claim that
  \begin{equation}
    \label{ind1}
    Z_{t,s}^2 \le C ( \mathcal{N} \otimes I + I \otimes \mathcal{N}  + N_1 + N_2 )^2
  \end{equation}
  for some constant $C$.
  In fact, by Cauchy-Schwarz inequality, we have
  \[
    \langle \psi, Z_{t,s}^2 \psi \rangle \le 3(A + B),
  \]
  where
  \begin{align*}
    A & = \langle \psi, (\mathcal{N} \otimes I + I \otimes \mathcal{N} + N_1 + N_2)^2 \psi \rangle, \\
    B & = \langle \psi, e^{i(t-s) \mathcal{H}} (\sqrt{N_1} \phi(u_t) \otimes I + \sqrt{N_2} I \otimes \phi(v_t) )^2 e^{-i(t-s) \mathcal{H}} \psi \rangle.
  \end{align*}
  Using the Cauchy-Schwarz inequality and Lemma \ref{l:relbN}, it is simple to obtain
  \[
    B \le 24 \langle \psi, (\mathcal{N} \otimes I + I \otimes \mathcal{N} + N_1 + N_2)^2 \psi \rangle.
  \]
  These estimates for $A$ and $B$ prove the claim.

  Now, using the notation $\mathrm{ad}_A(B) = [B,A]$, we calculate
  \[
    \begin{split}
      \mathrm{ad}_{Z_{t,s}}^m (\mathcal{N} \otimes I + I \otimes \mathcal{N}) & = e^{i(t-s) \mathcal{H}} \big( \sqrt{N_1} (a(u_t) + (-1)^m a^*(u_t) ) \otimes I \\
      & \qquad \qquad + \sqrt{N_2} I \otimes (a(v_t) + (-1)^m a^*(v_t)) \big) e^{-i(t-s) \mathcal{H}} \\
      & \quad - (1+(-1)^m) (N_2 + N_2).
    \end{split}
  \]
  By performing a long but straightforward calculation, we prove that
  \begin{equation}
    \label{ind2}
    \mathrm{ad}_Z^m(\mathcal{N} \otimes I + I \otimes \mathcal{N}) \mathrm{ad}_Z^m(\mathcal{N} \otimes I + I \otimes \mathcal{N})^* \le C ( \mathcal{N} \otimes I + I \otimes \mathcal{N} + N_1 + N_2 )^2
  \end{equation}
  for all $m \in \mathbb{N}$.

  By induction, it follows that, for all $j \in \mathbb{N}$, there are constants $C_j$ and $D_j$ such that
  \begin{equation}
    \label{Z1}
    Z_{t,s}^{j-1} (\mathcal{N} \otimes I + I \otimes \mathcal{N} + N_1 + N_2)^2 Z_{t,s}^{j-1} \le C_j (\mathcal{N} \otimes I + I \otimes \mathcal{N} + N_1 + N_2)^{2j},
  \end{equation}
  \begin{equation}
    \label{Z2}
    Z_{t,s}^{2j} \le D_j (\mathcal{N} \otimes I + I \otimes \mathcal{N} + N_1 + N_2)^{2j}.
  \end{equation}
  In fact, this is already proved for $j=1$ by \eqref{ind1} and \eqref{ind2}.
  Suppose that \eqref{Z1} and \eqref{Z2} hold true for all $j < k$.
  We will prove the estimates for $j = k$.

  For \eqref{Z1}, we have
  \[
    \begin{split}
      & Z_{t,s}^{k-1} ( \mathcal{N} \otimes I + I \otimes \mathcal{N} + N_1 + N_2 )^2 Z_{t,s}^{k-1} \\
      & = (\mathcal{N} \otimes I + I \otimes \mathcal{N} + N_1 + N_2) Z_{t,s}^{k-1} (\mathcal{N} \otimes I + I \otimes \mathcal{N} + N_1 + N_2) Z_{t,s}^{k-1} \\
      & \quad + [Z_{t,s}^{k-1}, \mathcal{N} \otimes I + I \otimes \mathcal{N} + N_1 + N_2] (\mathcal{N} \otimes I + I \otimes \mathcal{N} + N_1 + N_2) Z_{t,s}^{k-1} \\
      & \le 4 (\mathcal{N} \otimes I + I \otimes \mathcal{N} + N_1 + N_2) Z_{t,s}^{2k-2} (\mathcal{N} \otimes I + I \otimes \mathcal{N} + N_1 + N_2) \\
      & \quad + \tfrac{1}{2} Z_{t,s}^{k-1} ( \mathcal{N} \otimes I + I \otimes \mathcal{N} + N_1 + N_2 )^2 Z_{t,s}^{k-1} \\
      & \quad + 4 [Z_{t,s}^{k-1}, \mathcal{N} \otimes I + I \otimes \mathcal{N} + N_1 + N_2] \, [Z_{t,s}^{k-1}, \mathcal{N} \otimes I + I \otimes \mathcal{N} + N_1 + N_2]^*.
    \end{split}
  \]
  Thus
  \[
    \begin{split}
      & Z_{t,s}^{k-1} ( \mathcal{N} \otimes I + I \otimes \mathcal{N} + N_1 + N_2 )^2 Z_{t,s}^{k-1} \\
      & \le 8 (\mathcal{N} \otimes I + I \otimes \mathcal{N} + N_1 + N_2) Z_{t,s}^{2k-2} (\mathcal{N} \otimes I + I \otimes \mathcal{N} + N_1 + N_2) \\
      & \quad + 8 [Z_{t,s}^{k-1}, \mathcal{N} \otimes I + I \otimes \mathcal{N} + N_1 + N_2] \, [Z_{t,s}^{k-1}, \mathcal{N} \otimes I + I \otimes \mathcal{N} + N_1 + N_2]^* \\
      & \le C_k (\mathcal{N} \otimes I + I \otimes \mathcal{N} + N_1 + N_2)^{2k},
    \end{split}
  \]
  as desired.
  Here, we used the operator inequality
  \[
    \begin{split}
      & [Z_{t,s}^{k-1}, \mathcal{N} \otimes I + I \otimes \mathcal{N} + N_1 + N_2] \, [Z_{t,s}^{k-1}, \mathcal{N} \otimes I + I \otimes \mathcal{N} + N_1 + N_2]^* \\
      & \le (\mathcal{N} \otimes I + I \otimes \mathcal{N} + N_1 + N_2)^{2k},
    \end{split}
  \]
  which follows using the commutator expansion
  \[
    [Z_{t,s}^{k-1}, \mathcal{N} \otimes I + I \otimes \mathcal{N}] = \sum_{m=0}^{k-2} \binom{k-1}{m} Z_{t,s}^m \mathrm{ad}_{Z_{t,s}}^{k-1-m}(\mathcal{N} \otimes I + I \otimes \mathcal{N}).
  \]

  For \eqref{Z2}, using \eqref{ind1} and \eqref{Z1}, we obtain
  \[
    \begin{split}
      Z_{t,s}^{2k} & = Z_{t,s}^{k-1} Z_{t,s}^2 Z_{t,s}^{k-1} \\
      & \le C Z_{t,s}^{k-1} (\mathcal{N} \otimes I + I \otimes \mathcal{N} + N_1 + N_2)^2 Z_{t,s}^{k-1} \\
      & \le C C_j (\mathcal{N} \otimes I + I \otimes \mathcal{N} + N_1 + N_2)^{2k},
    \end{split}
  \]
  as desired.
  Therefore, we have proved \eqref{Z1} and \eqref{Z2}.

  Let us finish the proof of \eqref{step2b}.
  First we observe that, similarly as we estimated $Z_{t,s}^{2j}$, we can prove that
  \begin{multline*}
    (\mathcal{N} \otimes I + I \otimes \mathcal{N} + \sqrt{N_1} (\phi(u_t) \otimes I) + \sqrt{N_1} (I \otimes \phi(v_t)) + 2(N_1 + N_2))^{2j} \\
    \le C_j (\mathcal{N} \otimes I + I \otimes \mathcal{N} + N_1 + N_2)^{2j}
  \end{multline*}
  for all $j \in \mathbb{N}$.
  Hence
  \[
    \begin{split}
      & \langle \psi, U(t,s)^* (\mathcal{N} \otimes I + I \otimes \mathcal{N})^{2j} U(t,s) \psi \rangle \\
      & = \langle \mathcal{W}_t \psi, Z_{t,s}^{2j} \mathcal{W}_t \psi \rangle \\
      & \le C_j \langle \mathcal{W}_t \psi, (\mathcal{N} \otimes I + I \otimes \mathcal{N} + N_1 + N_2)^{2j} \mathcal{W}_t \psi \rangle \\
      & = C_j \langle \psi, (\mathcal{N} \otimes I + I \otimes \mathcal{N} + \sqrt{N_1} (\phi(u_t) \otimes I) + \sqrt{N_1} (I \otimes \phi(v_t)) + 2(N_1 + N_2))^{2j} \psi \rangle \\
      & \le C_j \langle \psi, (\mathcal{N} \otimes I + I \otimes \mathcal{N} + N_1 + N_2)^{2j} \psi \rangle.
    \end{split}
  \]
  This proves \eqref{step2b}.

  Finally, we derive \eqref{step2c} using \eqref{step2b}:
  \[
    \begin{split}
      & \langle \psi, U(t,s)^* (\mathcal{N} \otimes I + I \otimes \mathcal{N})^{2j+1} U(t,s) \psi \rangle \\
      & = \left\langle \frac{(\mathcal{N} \otimes I + I \otimes \mathcal{N})^{j+1}}{\sqrt{N_1 + N_2}} U(t,s) \psi, \sqrt{N_1 + N_2} (\mathcal{N} \otimes I + I \otimes \mathcal{N})^j U(t,s) \psi \right\rangle \\
      & \le \frac{1}{N_1+N_2} \langle \psi, U(t,s)^* (\mathcal{N} \otimes I + I \otimes \mathcal{N})^{2j+2} U(t,s) \psi \rangle \\
      & \quad + (N_1+N_2) \langle \psi, U(t,s) (\mathcal{N} \otimes I + I \otimes \mathcal{N})^{2j} U(t,s) \psi \rangle \\
      & \le \frac{C_{j+1}}{N_1+N_2} \langle \psi, (\mathcal{N} \otimes I + I \otimes \mathcal{N} + N_1 + N_2)^{2j+2} \psi \rangle \\
      & \quad + (N_1+N_2) C_j \langle \psi, (\mathcal{N} \otimes I + I \otimes \mathcal{N} + N_1 + N_2)^{2j} \psi \rangle \\
      & \le D_j \langle \psi, (\mathcal{N} \otimes I + I \otimes \mathcal{N} + N_1 + N_2)^{2j+1}(\mathcal{N} \otimes I + I \otimes \mathcal{N} + 1) \psi \rangle
    \end{split}
  \]
  for some constants $C_j$ and $D_j$.
  The proof of Lemma \ref{l:step2} is complete.
\end{proof}

\textbf{Comparison of the $U$ and $U^c$ dynamics.}
Now that we have a-priori estimates for the number of particles with respect to the fluctuation dynamics $U$, we can compare the evolution $U_{t,s}$ with the evolution $U_{t,s}^c$.

\begin{lemma}
  \label{l:step3}
  For every $j \in \mathbb{N}$, there exist constants $C_j$ and $K_j$ such that
  \begin{equation}
    \label{step3a}
    \begin{split}
      & |\langle U(t,s) \psi, (\mathcal{N} \otimes I + I \otimes \mathcal{N})^j (U(t,s) - U^c(t,s)) \psi \rangle| \\
      & \le C_j \left[ \left( \frac{N_1+N_2}{M_1} \right)^j + \left( \frac{N_1+N_2}{M_2} \right)^j \right] \| ((\mathcal{N}+1) \otimes I + I \otimes (\mathcal{N}+1))^{j+1} \psi \|^2 \\
      & \quad \qquad \times \frac{\exp(K_j (1 + \sqrt{M_1/N_1} + \sqrt{M_2/N_2}) |t-s|)}{1 + \sqrt{M_1/N_1} + \sqrt{M_2/N_2}}
    \end{split}
  \end{equation}
  and
  \begin{equation}
    \label{step3b}
    \begin{split}
      & |\langle U^c(t,s) \psi, (\mathcal{N} \otimes I + I \otimes \mathcal{N})^j (U(t,s) - U^c(t,s)) \psi \rangle| \\
      & \le C_j \left[ \frac{1}{M_1^j} + \frac{1}{M_2^j} \right] \| ((\mathcal{N}+1) \otimes I + I \otimes (\mathcal{N}+1))^{j+1} \psi \|^2 \\
      & \quad \qquad \times \frac{\exp(K_j (1 + \sqrt{M_1/N_1} + \sqrt{M_2/N_2}) |t-s|)}{1 + \sqrt{M_1/N_1} + \sqrt{M_2/N_2}}
    \end{split}
  \end{equation}
  for $\psi \in \mathcal{F} \otimes \mathcal{F}$ and $t,s \in \R$.
\end{lemma}

\begin{proof}
  To prove \eqref{step3a}, we use the identity
  \[
    U(t,s) - U^c(t,s) = -i \int_s^t dr \, U(t,r)(\mathcal{L}(r) - \mathcal{L}^c(r)) U^c(r,s)
  \]
  and observe that
  \[
    \begin{split}
      \mathcal{L}(t) - \mathcal{L}^c(t) & = C_{N_1,N_2}(t) - C^c_{N_1,N_2}(t) \\
      & = \frac{1}{\sqrt{N_1}} \int dx dz \, V_1(x-z) b_x^* (u_t(z) \chi_1^c b_z^* + \overline{u_t}(z) b_z \chi_1^c) b_x \\
      & \quad + \frac{1}{\sqrt{N_2}} \int dy dz \, V_2(y-z) c_y^* (v_t(z) \chi_2^c c_z^* + \overline{v_t}(z) c_z \chi_2^c) c_y \\
      & \quad + \frac{\sqrt{N_1}}{N_1 + N_2} \int dx dy \, V_{12}(x-y) c_y^* ( u_t(x) \chi_1^c b_x^* + \overline{u_t}(x) b_x \chi_1^c ) c_y \\
      & \quad + \frac{\sqrt{N_2}}{N_1 + N_2} \int dx dy \, V_{12}(x-y) b_x^* ( v_t(y) \chi_2^c c_y^* + \overline{v_t}(y) c_y \chi_2^c ) b_x
    \end{split}
  \]
  where
  \begin{align*}
    \chi_1^c & = \chi(\mathcal{N} > M_1) \otimes I \\
    \chi_2^c & = I \otimes \chi(\mathcal{N} > M_2).
  \end{align*}
  Using these formulae, we write
  \begin{equation}
    \label{expw}
    \langle U(t,s) \psi, (\mathcal{N} \otimes I + I \otimes \mathcal{N})^j (U(t,s) - U^c(t,s)) \psi \rangle = \sum_{j=1}^8 w_j
  \end{equation}
  with
  \begin{align*}
    w_1 & = \frac{-i}{\sqrt{N_1}} \int_s^t dr \int dx \langle b_x U_{t,s}^* (\mathcal{N} \otimes I + I \otimes \mathcal{N})^j U_{t,s} \psi, b(V_1(x-\cdot) u_t(\cdot)) \chi_1^c b_x U^c_{t,s} \psi \rangle, \\
    w_2 & = \frac{-i}{\sqrt{N_1}} \int_s^t dr \int dx \langle b_x U_{t,s}^* (\mathcal{N} \otimes I + I \otimes \mathcal{N})^j U_{t,s} \psi, \chi_1^c b^*(V_1(x-\cdot) u_t(\cdot)) b_x U^c_{t,s} \psi \rangle, \\
    w_3 & = \frac{-i}{\sqrt{N_2}} \int_s^t dr \int dy \langle c_y U_{t,s}^* (\mathcal{N} \otimes I + I \otimes \mathcal{N})^j U_{t,s} \psi, c(V_2(y-\cdot) v_t(\cdot)) \chi_2^c c_y U^c_{t,s} \psi \rangle, \\
    w_4 & = \frac{-i}{\sqrt{N_2}} \int_s^t dr \int dy \langle c_y U_{t,s}^* (\mathcal{N} \otimes I + I \otimes \mathcal{N})^j U_{t,s} \psi, \chi_2^c c^*(V_2(y-\cdot) v_t(\cdot)) c_y U^c_{t,s} \psi \rangle, \\
    w_5 & = \alpha_1 \int_s^t dr \int dy \langle c_y U_{t,s}^* (\mathcal{N} \otimes I + I \otimes \mathcal{N})^j U_{t,s} \psi, c(V_2(y-\cdot) v_t(\cdot)) \chi_2^c c_y U^c_{t,s} \psi \rangle, \\
    w_6 & = \alpha_1 \int_s^t dr \int dy \langle c_y U_{t,s}^* (\mathcal{N} \otimes I + I \otimes \mathcal{N})^j U_{t,s} \psi, \chi_2^c c^*(V_2(y-\cdot) v_t(\cdot)) c_y U^c_{t,s} \psi \rangle, \\
    w_7 & = \alpha_2 \int_s^t dr \int dx \langle b_x U_{t,s}^* (\mathcal{N} \otimes I + I \otimes \mathcal{N})^j U_{t,s} \psi, b(V_1(x-\cdot) u_t(\cdot)) \chi_1^c b_x U^c_{t,s} \psi \rangle, \\
    w_8 & = \alpha_2 \int_s^t dr \int dx \langle b_x U_{t,s}^* (\mathcal{N} \otimes I + I \otimes \mathcal{N})^j U_{t,s} \psi, \chi_1^c b^*(V_1(x-\cdot) u_t(\cdot)) b_x U^c_{t,s} \psi \rangle,
  \end{align*}
  where
  \[
    \alpha_j = \frac{-i \sqrt{N_j}}{N_1+N_2}.
  \]
  Notice that $w_j = w_j(t)$ for $j = 1, \dots, 8$.
  Hence, to estimate the absolute value of the left hand side of \eqref{expw}, we need to estimate $|w_j|$ for each $j$.

  Let us estimate $|w_1|$.
  We will denote by $\chi_j^c(Q)$ the operator $\chi_j^c$ with $M_j$ replaced by $Q$.
  Calculating, we obtain
  \[
    \begin{split}
      & |w_1| \\
      & \le \frac{1}{\sqrt{N_1}} \int_s^t dr \int dx \, \| b_x U_{t,s}^* (\mathcal{N} \otimes I + I \otimes \mathcal{N})^j U_{t,s} \psi \| \\
      & \qquad \qquad \qquad \qquad \times \| b(V_1(x-\cdot) u_t(\cdot)) b_x \chi_1^c(M_1+1) U^c_{t,s} \psi \| \\
      & \le \frac{1}{\sqrt{N_1}} \sup_x \| V_1(x-\cdot) u_t(\cdot) \| \int_s^t dr \int dx \, \| b_x U_{t,s}^* (\mathcal{N} \otimes I + I \otimes \mathcal{N})^j U_{t,s} \psi \| \\
      & \qquad \qquad \qquad \qquad \times \| b_x ((\mathcal{N}-1)^{1/2} \otimes I) \chi_1^c(M_1+1) U^c_{t,s} \psi \| \\
      & \le \frac{C}{\sqrt{N_1}} \int_s^t dr \, \| (\mathcal{N}^{1/2} \otimes I) U_{t,s}^* (\mathcal{N} \otimes I + I \otimes \mathcal{N})^j U_{t,s} \psi \| \, \| (\mathcal{N} \otimes I) \chi_1^c U^c_{t,s} \psi \|.
    \end{split}
  \]
  Furthermore, using Lemma \ref{l:step2},
  \begin{equation}
    \label{est1}
    \begin{split}
      & \| (\mathcal{N}^{1/2} \otimes I) U_{t,s}^* (\mathcal{N} \otimes I + I \otimes \mathcal{N})^j U_{t,s} \psi \|^2 \\
      & = \langle (\mathcal{N} \otimes I + I \otimes \mathcal{N})^j U_{t,s} \psi, U_{t,s} (\mathcal{N} \otimes I) U_{t,s}^* (\mathcal{N} \otimes I + I \otimes \mathcal{N})^j U_{t,s} \psi \rangle \\
      & \le 24 \langle (\mathcal{N} \otimes I + I \otimes \mathcal{N})^j U_{t,s} \psi, \\
      & \qquad \qquad (\mathcal{N} \otimes I + I \otimes \mathcal{N} + N_1 + N_2 + 1) (\mathcal{N} \otimes I + I \otimes \mathcal{N})^j U_{t,s} \psi \rangle \\
      & \le 24 \langle U_{t,s} \psi, (\mathcal{N} \otimes I + I \otimes \mathcal{N})^{2j+1} U_{t,s} \psi \rangle \\
      & \quad + 24 (N_1 + N_2 + 1) \langle U_{t,s} \psi, (\mathcal{N} \otimes I + I \otimes \mathcal{N})^{2j} U_{t,s} \psi \rangle \\
      & \le 24 D_j \langle \psi, (\mathcal{N} \otimes I + I \otimes \mathcal{N})^{2j+1} (\mathcal{N} \otimes I + I \otimes \mathcal{N} + 1) \psi \rangle \\
      & \quad + 24(N_1 + N_2 + 1) C_j \langle \psi, (\mathcal{N} \otimes I + I \otimes \mathcal{N} + N_1 + N_2)^{2j} \psi \rangle \\
      & \le \tilde{D}_j (N_1 + N_2)^{2j+1} \langle \psi, (\mathcal{N} \otimes I + I \otimes \mathcal{N} + 1)^{2j+2} \psi \rangle \\
      & \quad + \tilde{C}_j (N_1 + N_2)^{2j+1} \langle \psi, (\mathcal{N} \otimes I + I \otimes \mathcal{N})^{2j} \psi \rangle \\
      & \le C_j (N_1 + N_2)^{2j+1} \langle \psi, ((\mathcal{N}+1) \otimes I + I \otimes (\mathcal{N}+1))^{2j+2} \psi \rangle
    \end{split}
  \end{equation}
  and
  \[
    \begin{split}
      & \| (\mathcal{N} \otimes I) \chi_1^c U^c_{t,s} \psi \|^2 \\
      & = \langle (\mathcal{N} \otimes I) U^c_{t,s} \psi, \chi_1^c (\mathcal{N} \otimes I) U_{t,s}^c \psi \rangle \\
      & \le \left\langle (\mathcal{N} \otimes I) U^c_{t,s} \psi, \left( \frac{\mathcal{N} \otimes I}{M_1} \right)^{2j} (\mathcal{N} \otimes I) U_{t,s}^c \psi \right\rangle \\
      & = M_1^{-2j} \langle U^c_{t,s} \psi, (\mathcal{N} \otimes I)^{2j+2} U^c_{t,s} \psi \rangle \\
      & \le M_1^{-2j} \langle U^c_{t,s} \psi, (\mathcal{N} \otimes I + I \otimes \mathcal{N})^{2j+2} U^c_{t,s} \psi \rangle,
    \end{split}
  \]
  where we used the operator inequality
  \[
    \chi(\mathcal{N} > M_1) \le (\mathcal{N}/M_1)^{2j}.
  \]
  Hence
  \[
    \begin{split}
      & |w_1| \\
      & \le C_j \left( \frac{N_1+N_2}{M_1} \right)^j \left( \frac{N_1+N_2}{N_1} \right)^{1/2} \\
      & \quad \times \| ((\mathcal{N}+1) \otimes I + I \otimes (\mathcal{N}+1))^{j+1} \psi \| \\
      & \quad \times \int_s^t dr \, \langle U^c_{t,s} \psi, (\mathcal{N} \otimes I + I \otimes \mathcal{N})^{2j+2} U^c_{t,s} \psi \rangle^{1/2} \\
      & \le C_j \left( \frac{N_1+N_2}{M_1} \right)^j \| ((\mathcal{N}+1) \otimes I + I \otimes (\mathcal{N}+1))^{j+1} \psi \|^2 \\
      & \quad \qquad \times \int_s^t dr \, \exp(K_j(1 + \sqrt{M_1/N_1} + \sqrt{M_2/N_2}) (r-s)) \\
      & \le C_j \left( \frac{N_1+N_2}{M_1} \right)^j \| ((\mathcal{N}+1) \otimes I + I \otimes (\mathcal{N}+1))^{j+1} \psi \|^2 \\
      & \quad \qquad \times \frac{\exp(K_j (1 + \sqrt{M_1/N_1} + \sqrt{M_2/N_2}) |t-s|)}{1 + \sqrt{M_1/N_1} + \sqrt{M_2/N_2}},
    \end{split}
  \]
  where we used Lemma \ref{l:step1} and Lemma \ref{l:Z}.

  Now, we estimate $|w_5|$.
  Calculating, we obtain
  \[
    \begin{split}
      & |w_5| \\
      & \le \frac{\sqrt{N_1}}{N_1+N_2} \int_s^t dr \int dy \, \| c_y U_{t,s}^* (\mathcal{N} \otimes I + I \otimes \mathcal{N})^j U_{t,s} \psi \| \\
      & \qquad \qquad \qquad \qquad \times \| b(V_{12}(y-\cdot) u_t(\cdot)) c_y \chi_2^c U^c_{t,s} \psi \| \\
      & \le \frac{\sqrt{N_1}}{N_1+N_2} \sup_y \| V_{12}(y-\cdot) u_t(\cdot) \| \int_s^t dr \int dy \, \| c_y U_{t,s}^* (\mathcal{N} \otimes I + I \otimes \mathcal{N})^j U_{t,s} \psi \| \\
      & \qquad \qquad \qquad \qquad \times \| c_y (\mathcal{N}^{1/2} \otimes I) \chi_1^c(M_1+1) U^c_{t,s} \psi \| \\
      & \le \frac{C\sqrt{N_1}}{N_1+N_2} \int_s^t dr \, \| (I \otimes \mathcal{N}^{1/2}) U_{t,s}^* (\mathcal{N} \otimes I + I \otimes \mathcal{N})^j U_{t,s} \psi \| \\
      & \qquad \qquad \qquad \qquad \times \| ((\mathcal{N}+1)^{1/2} \otimes (\mathcal{N}+1)^{1/2}) \chi_1^c U^c_{t,s} \psi \|.
    \end{split}
  \]
  Similarly as above, we prove that
  \begin{multline*}
    \| (I \otimes \mathcal{N}^{1/2}) U_{t,s}^* (\mathcal{N} \otimes I + I \otimes \mathcal{N})^j U_{t,s} \psi \|^2 \\
    \le C_j (N_1 + N_2)^{2j+1} \langle U_{t,s} \psi, ((\mathcal{N}+1) \otimes I + I \otimes (\mathcal{N}+1))^{2j+2} U_{t,s} \psi \rangle
  \end{multline*}
  and
  \[
    \begin{split}
      & \| (\mathcal{N}+1)^{1/2} \otimes (\mathcal{N}+1)^{1/2} \chi_1^c U^c_{t,s} \psi \|^2 \\
      & = \langle (\mathcal{N}+1)^{1/2} \otimes (\mathcal{N}+1)^{1/2} U^c_{t,s} \psi, \chi_1^c (\mathcal{N}+1)^{1/2} \otimes (\mathcal{N}+1)^{1/2} U_{t,s}^c \psi \rangle \\
      & \le \left\langle (\mathcal{N}+1)^{1/2} \otimes (\mathcal{N}+1)^{1/2} U^c_{t,s} \psi, \left( \frac{\mathcal{N} \otimes I}{M_1} \right)^{2j} (\mathcal{N}+1)^{1/2} \otimes (\mathcal{N}+1)^{1/2} U_{t,s}^c \psi \right\rangle \\
      & = M_1^{-2j} \langle U^c_{t,s} \psi, (\mathcal{N} \otimes I)^{2j+2} U^c_{t,s} \psi \rangle \\
      & \le M_1^{-2j} \langle U^c_{t,s} \psi, (\mathcal{N} \otimes I + I \otimes \mathcal{N})^{2j+2} U^c_{t,s} \psi \rangle.
    \end{split}
  \]
  Thus, since $\sqrt{N_1/(N_1+N_2)} < C$ by Lemma \ref{l:Z}, similarly as above, we have
  \[
    \begin{split}
      & |w_5| \\
      & \le C_j \left( \frac{N_1+N_2}{M_1} \right)^j \| ((\mathcal{N}+1) \otimes I + I \otimes (1+\mathcal{N}))^{j+1} \psi \|^2 \\
      & \quad \qquad \times \frac{\exp(K_j (1 + \sqrt{M_1/N_1} + \sqrt{M_2/N_2}) |t-s|)}{1 + \sqrt{M_1/N_1} + \sqrt{M_2/N_2}}.
    \end{split}
  \]
  Similarly as we estimated $|w_1|$ and $|w_5|$, we estimate $|w_j|$ for $j = 2,3,4,6,7,8$ and obtain similar bounds with $M_1$ replaced by $M_2$ accordingly.
  This proves \eqref{step3a}.

  To prove \eqref{step3b}, we proceed similarly as in the proof of \eqref{step3a}.
  Calculating, we obtain
  \begin{equation}
    \label{expm}
    \langle U^c(t,s) \psi, (\mathcal{N} \otimes I + I \otimes \mathcal{N})^j (U(t,s) - U^c(t,s)) \psi \rangle = \sum_{j=1}^8 m_j,
  \end{equation}
  where $m_j$ for $j=1,\dots,8$ are defined exactly as $w_j$ for $j=1,\dots,8$, but with $U(t,s)$ on the left side replaced by $U^c(t,s)$.
  We want to estimate $|m_j|$ for each $j$.
  The calculations are very similar to the calculations for estimating $|w_j|$, except that instead of \eqref{est1} we have
  \[
    \begin{split}
      & \| (\mathcal{N}^{1/2} \otimes I) U_{t,s}^* (\mathcal{N} \otimes I + I \otimes \mathcal{N})^j U^c_{t,s} \psi \|^2 \\
      & = \langle (\mathcal{N} \otimes I + I \otimes \mathcal{N})^j U^c_{t,s} \psi, U_{t,s} (\mathcal{N} \otimes I) U_{t,s}^* (\mathcal{N} \otimes I + I \otimes \mathcal{N})^j U^c_{t,s} \psi \rangle \\
      & \le 24 \langle (\mathcal{N} \otimes I + I \otimes \mathcal{N})^j U^c_{t,s} \psi, \\
      & \qquad \qquad \qquad (\mathcal{N} \otimes I + I \otimes \mathcal{N} + N_1 + N_2 + 1) (\mathcal{N} \otimes I + I \otimes \mathcal{N})^j U^c_{t,s} \psi \rangle \\
      & \le 24 \langle U^c_{t,s} \psi, (\mathcal{N} \otimes I + I \otimes \mathcal{N})^{2j+1} U^c_{t,s} \psi \rangle \\
      & \quad + 24 (N_1 + N_2 + 1) \langle U^c_{t,s} \psi, (\mathcal{N} \otimes I + I \otimes \mathcal{N})^{2j} U^c_{t,s} \psi \rangle \\
      & \le C_j (N_1 + N_2) \| ( (\mathcal{N}+1) \otimes I + I \otimes (\mathcal{N}+1) )^{j+1} \|^2 \\
      & \quad \qquad \times \exp(K_j (1 + \sqrt{M_1/N_1} + \sqrt{M_2/N_2}) |t-s|).
    \end{split}
  \]
  Notice the absence of a factor $(N_1 + N_2)^{2j}$ in the numerator.
  By taking this difference into account, we conclude that each $|m_j|$ is bounded by
  \[
    \frac{C_j}{M_k^j} \| ((\mathcal{N}+1) \otimes I + I \otimes (\mathcal{N}+1))^{j+1} \psi \|^2 \frac{\exp(K_j (1 + \sqrt{M_1/N_1} + \sqrt{M_2/N_2}) |t-s|)}{1 + \sqrt{M_1/N_1} + \sqrt{M_2/N_2}},
  \]
  where $M_k$ is either $M_1$ or $M_2$, accondingly.
  This completes the proof of \eqref{step3b}.
  We have proved Lemma \ref{l:step3}.
\end{proof}

We are ready to prove Proposition \ref{p:part2}.

\begin{proof}[Proof of Proposition \ref{p:part2}.]
  Write
  \[
    \begin{split}
      \langle U_{t,s} \psi, (\mathcal{N} \otimes I + I \otimes \mathcal{N})^j U_{t,s} \psi \rangle & = \langle U_{t,s} \psi, (\mathcal{N} \otimes I + I \otimes \mathcal{N})^j (U_{t,s} - U^c_{t,s}) \psi \rangle \\
      & \quad + \langle (U_{t,s} - U^c_{t,s}) \psi, (\mathcal{N} \otimes I + I \otimes \mathcal{N})^j U^c_{t,s} \psi \rangle \\
      & \quad + \langle U^c_{t,s} \psi, (\mathcal{N} \otimes I + I \otimes \mathcal{N})^j U^c_{t,s} \psi \rangle.
    \end{split}
  \]
  Using Lemmata \ref{l:step1} and \ref{l:step3} with $M_1 = N_1$ and $M_2 = N_2$, and Lemma \ref{l:Z}, we obtain
  \[
    \langle U_{t,s} \psi, (\mathcal{N} \otimes I + I \otimes \mathcal{N})^j U_{t,s} \psi \rangle \le C_j \| [(\mathcal{N} + 1) \otimes I + I \otimes (\mathcal{N}+1)]^{j+1} \psi \|^2 e^{\gamma_j |t-s|}
  \]
  for some constants $C_j$ and $\gamma_j$.
  This proves Proposition \ref{p:part2}.
\end{proof}

\end{document}